\documentclass[12pt,american,pra,reprint]{revtex4-2}
\usepackage[T1]{fontenc}
\usepackage[latin9]{inputenc}
\usepackage{color}
\usepackage{babel}
\usepackage{mathrsfs}
\usepackage{amsmath}
\usepackage{amsthm}
\usepackage{amssymb}
\usepackage{graphicx}
\usepackage[unicode=true,pdfusetitle,
 bookmarks=true,bookmarksnumbered=false,bookmarksopen=false,
 breaklinks=false,pdfborder={0 0 0},pdfborderstyle={},backref=false,colorlinks=true]
 {hyperref}

\makeatletter
\theoremstyle{plain}
\newtheorem{thm}{\protect\theoremname}
\theoremstyle{definition}
\newtheorem{defn}[thm]{\protect\definitionname}
\theoremstyle{remark}
\newtheorem{rem}[thm]{\protect\remarkname}

\usepackage{bbm}

\makeatother

\providecommand{\definitionname}{Definition}
\providecommand{\remarkname}{Remark}
\providecommand{\theoremname}{Theorem}

\begin{document}
\title{Dynamical Resources}
\author{Gilad Gour}
\email{gour@ucalgary.ca}

\affiliation{Department of Mathematics and Statistics, University of Calgary, AB,
Canada T2N 1N4}
\affiliation{Institute for Quantum Science and Technology, University of Calgary,
AB, Canada T2N 1N4}
\author{Carlo Maria Scandolo}
\email{carlomaria.scandolo@ucalgary.ca}

\affiliation{Department of Mathematics and Statistics, University of Calgary, AB,
Canada T2N 1N4}
\affiliation{Institute for Quantum Science and Technology, University of Calgary,
AB, Canada T2N 1N4}
\begin{abstract}
Quantum channels are quintessential to quantum information, being
used in all protocols, and describing how systems evolve in space
and time. As such, they play a key role in the manipulation of quantum
resources, and they are often resources themselves, called dynamical
resources. This forces us to go beyond standard resource theories
of quantum states. Here we provide a rigorous foundation for dynamical
resource theories, where the resources into play are quantum channels,
explaining how to manipulate dynamical resources with free superchannels.
In particular, when the set of free superchannels is convex, we present
a novel construction of an infinite and complete family of convex
resource monotones, giving necessary and sufficient conditions for
convertibility under free superchannels. After showing that the conversion
problem in convex dynamical resource theories can be solved with conic
linear programming, we define various resource-theoretic protocols
for dynamical resources. These results serve as the framework for
the study of concrete examples of theories of dynamical resources,
such as dynamical entanglement theory.
\end{abstract}
\maketitle

\section{Introduction}

The remarkable success of quantum information stems from the fact
that quantum objects provide concrete advantages in several tasks.
Think, for instance, of entangled states \citep{Review-entanglement},
which can be harnessed to implement protocols that have no classical
analogue \citep{Teleportation,Dense-coding,Ekert,Entanglement-swapping}.
Similar to entanglement, other quantum features are resources, such
as coherence in quantum superpositions \citep{Review-coherence}.
The idea of entanglement and other quantum features helping in information-theoretic
tasks can be made rigorous with the framework of\emph{ resource theories}
\citep{Quantum-resource-1,Quantum-resource-2,Resource-knowledge,Resource-currencies,Gour-single-shot,Regula2017,Multiresource,Gour-review,Adesso-resource,Single-shot-new,Kuroiwa2020generalquantum}.
This framework is so general and powerful that it can be extended
even beyond the quantum case \citep{Resource-theories,Resource-monoid,Chiribella-Scandolo-entanglement,Fong,TowardsThermo,EastThesis,Purity,ScandoloPhD,Takagi-Regula,Resource-GPT-infinite}. 

Resource theories have been used to study a great number of physical
situations \citep{Gour-review}, always providing new insights into
quantum theory and novel results for quantum information protocols.
The basic idea behind them is that an agent operates on a quantum
system to perform some task, but they do \emph{not} have access to
the full set of quantum operations. Instead, they can only perform
a strict subset of them, called \emph{free operations}. Similarly,
they cannot prepare the full set of quantum states, but only a strict
subset of them, the \emph{free states}. The restriction usually comes
from the physical constraints of the task the agent is trying to perform:
free operations are those that are easy to implement in the physical
scenario the agent operates in. Anything that can help the agent overcome
their restriction is regarded as a \emph{valuable resource}.

The convertibility between two resources under free operations sets
up a preorder on the set of resources, whereby a resource is more
valuable than another if the former can be converted into the latter
by some free operation. In simpler terms, a resource is more valuable
than another if, from the former, it is possible to reach a larger
set of resources. This allows one to introduce the notion of resource
monotone, a real-valued function that assigns a ``price'' to resources
according to their preorder. Monotones often have a very important
operational and physical meaning (e.g.\ the entropy or the free energy
in quantum thermodynamics \citep{delRio,Xuereb,Lostaglio-thermo}),
for they quantify how well a given task can be performed \citep{Gour-review}.
Two tasks that are particularly relevant in resource theories are
extracting the maximum amount of the maximal resource out of a generic
resource (\emph{distillation}), and minimizing the amount of the maximal
resource necessary to produce a given resource (\emph{cost}) \citep{Quantum-resource-1,Quantum-resource-2,Resource-currencies,EastThesis,Gour-review,Kuroiwa2020generalquantum}.
The distillation and cost of a state obey a Carnot-like inequality,
with the distillation always less than or equal to the cost \citep{Thermodynamic-entanglement}.

Resource theories have been studied in great detail when the resources
involved are \emph{states} (also known as static resources) \citep{Gour-review}.
In this case, one wants to study the conversion between states. This
is the usual setting in which, e.g., one studies entanglement theory
\citep{Review-entanglement,Plenio-review}.

Nevertheless, if one looks closely at the first examples where entanglement
proved to be a resource (e.g.\ quantum teleportation \citep{Teleportation}
and dense coding \citep{Dense-coding}), one notices they involve
the conversion of a state into a particular channel, i.e.\ a static
resource into a dynamical one \citep{Devetak-Winter,Resource-calculus}.
Therefore the need to go beyond conversion between static resources
is built in the very first protocol showing the value of quantum resources.
This is supported by the fact that in physics everything, including
a state, can be viewed as a \emph{dynamical resource} \citep{Chiribella2008,Chiribella-purification,hardy2011}.
Extending resource theories from states to channels \citep{Gour-review,Resource-channels-1,Resource-channels-2,Gour-Winter}
has recently gained considerable attention \citep{Resource-theories,Fong,Coherence-beyond-states,Pirandola-LOCC,Gour2018,Gour2018a,Li,Berta-cost,Wilde-cost,WW18,Coherence-beyond2,Rosset-resource,Egloff-resource,Thermal-capacity,Magic-channels,Wang-magic,Process-Markov,Gaurav,Takagi-communication,Wilde-entanglement,Dynamical-entanglement,Dynamical-coherence-entanglement,Wolfe2020,Schmid2020,Rosset-distributed,LOSR-nonlocality},
because of their relevance in a lot of information-theoretic situations
\citep{Circuit-architecture,Chiribella2008,Gour-review,Resource-channels-2}.
Moreover, since quantum channels represent the most general ways in
which a physical system evolves, for a more effective exploitation
of quantum resources, it is essential to understand how they are consumed
or produced by evolution.

In theories of dynamical resources, the agent converts different channels
by means of a restricted set of \emph{supermaps} \citep{Chiribella2008,Hierarchy-combs,Switch,Perinotti1,Perinotti2,Gour2018,Supermeasurements}.
In particular, we focus on supermaps that send quantum channels to
quantum channels. They are called \emph{superchannels}. They are not
just abstract entities, but they can be realized in a laboratory with
a pre-processing channel and a post-processing channel, connected
by a memory system \citep{Chiribella2008,Gour2018}. Clearly, if we
take the pre- and the post-processing of a superchannel to be free
channels (according to some resource theory of states), we have a
free superchannel \citep{Resource-theories}, which sends free channels
to free channels (even in a complete sense). This is the most common
approach to constructing free superchannels \citep{Resource-theories,Resource-channels-2}.

In this article, which is a companion to Ref.~\citep{Dynamical-entanglement},
we present the general framework of resource theories of quantum processes,
which constitutes the mathematical framework for our treatment of
dynamical entanglement announced in Ref.~\citep{Dynamical-entanglement}.
We note how for the largest class of free superchannels in a resource
theory, which are completely resource non-generating superchannels,
it is not clear if they can actually be realized in terms of free
pre- and post-processing, and we conjecture that this is not the case.

Then we turn to the conversion problem, showing two ways to solve
it in convex dynamical resource theories by means of a conic linear
program. In the first approach, we construct a complete family of
convex dynamical monotones, which give necessary and sufficient conditions
for convertibility under free superchannels. In the second approach,
solving the conversion problem becomes equivalent to calculating a
particular type of distance---the conversion distance---from one
channel to another.

Finally, we present the classic resource-theoretic protocols of cost
and distillation both in the single-shot and the asymptotic regime.
We note that for dynamical resources such protocols take a new twist
from their static counterpart, whereby various dynamical resources
can also be applied one after another (and not just in parallel) to
create an \emph{adaptive strategy} \citep{Adaptive-metrology,Pirandola-adaptive-metrology,Pirandola-LOCC,Kaur2017,Wilde-cost,Wilde-entanglement,Dynamical-entanglement}.

The article is organized as follows. In section~\ref{sec:Preliminaries},
we present basic facts on the formalism of superchannels, including
a new result on the uniqueness of a superchannel realization in terms
of pre- and post-processing. In the same section we give an overview
of quantum resource theories as well. Section~\ref{sec:Resource-processes}
is all devoted to the general formalism of resource theories for quantum
processes, with a new construction of a complete set of monotones,
and precise definitions of several conversion protocols. Conclusions
are drawn in section~\ref{sec:Conclusions}.

\section{Preliminaries\label{sec:Preliminaries}}

This section contains some basic notions to understand the rest of
this article. First we specify the notation we use, and then we move
to give a brief overview of the formalism used to manipulate quantum
channels, namely supermaps, superchannels, and combs. Here we also
prove a new result (theorem~\ref{thm:uniqueness}), concerning the
uniqueness of the realization of a superchannel in terms of quantum
channels. Finally we give a brief introduction to resource theories.

\subsection{Notation}

Physical systems and their corresponding Hilbert spaces will be denoted
by $A$, $B$, $C$, etc, where we will use the notation $AB$ to
mean $A\otimes B$. Dimensions will be denoted with vertical bars;
e.g.\ the dimension of system $A$ will be denoted by $\left|A\right|$.
The tilde symbol will be reserved to indicate a replica of a system.
For example, $\widetilde{A}$ denotes a replica of $A$, i.e.\ $\left|A\right|=\left|\widetilde{A}\right|$.
Density matrices acting on Hilbert spaces will be denoted by lowercase
Greek letter $\rho$, $\sigma$, etc, with one exception for the maximally
mixed state (i.e.\ the uniform state), which will be denoted by $u_{A}:=\frac{1}{\left|A\right|}I_{A}$.

The set of all bounded operators acting on system $A$ is denoted
by $\mathfrak{B}\left(A\right)$, the set of all Hermitian matrices
acting on $A$ by $\mathrm{Herm}\left(A\right)$, and the set of all
density matrices acting on system $A$ by $\mathfrak{D}\left(A\right)$.
Note that $\mathfrak{D}\left(A\right)\subset\mathrm{Herm}\left(A\right)\subset\mathfrak{B}\left(A\right)$.
We use the calligraphic letters $\mathcal{D}$, $\mathcal{E}$, $\mathcal{F}$,
etc.\ to denote quantum channels, reserving $\mathcal{V}$ to represent
an isometry map. The identity map on a system $A$ will be denoted
by $\mathsf{id}_{A}$. The set of all linear maps from $\mathfrak{B}\left(A\right)$
to $\mathfrak{B}\left(B\right)$ is denoted by $\mathfrak{L}\left(A\to B\right)$,
the set of all completely positive (CP) maps by $\mathrm{CP}\left(A\to B\right)$,
and the set of quantum channels by $\mathrm{CPTP}\left(A\to B\right)$.
Note that $\mathrm{CPTP}\left(A\to B\right)\subset\mathrm{CP}\left(A\to B\right)\subset\mathfrak{L}\left(A\to B\right)$.
$\mathrm{Herm}\left(A\to B\right)$ will denote the real vector space
of all Hermitian-preserving maps in $\mathfrak{L}\left(A\to B\right)$.
We will write $\mathcal{N}\geq0$ to mean that the map $\mathcal{N}\in\mathrm{Herm}\left(A\to B\right)$
is completely positive.

Since in this paper we focus on dynamical resources in the form of
quantum channels, unless otherwise specified, it will be convenient
to associate two subsystems $A_{0}$ and $A_{1}$ with every physical
system $A$, referring, respectively, to the input and output of the
resource. Hence, any physical system will be comprised of two subsystems
$A=\left(A_{0},A_{1}\right)$, even those representing a static resource,
in which case we simply have $\left|A_{0}\right|=1$. For simplicity,
we will denote a channel with a subscript $A$, e.g.\ $\mathcal{N}_{A}$,
to mean that it is an element of $\mathrm{CPTP}\left(A_{0}\to A_{1}\right)$.
Similarly, a bipartite channel in $\mathrm{CPTP}\left(A_{0}B_{0}\to A_{1}B_{1}\right)$
will be denoted by $\mathcal{N}_{AB}$. This notation makes the analogy
with bipartite states more transparent.

In this setting, when we consider $A=\left(A_{0},A_{1}\right)$, $B=\left(B_{0},B_{1}\right)$,
etc.\ comprised of input and output subsystems, the symbol $\mathfrak{L}\left(A\to B\right)$
refers to all linear maps from the vector space $\mathfrak{L}\left(A_{0}\to A_{1}\right)$
to the vector space $\mathfrak{L}\left(B_{0}\to B_{1}\right)$. Similarly,
$\mathrm{Herm}\left(A\to B\right)\subset\mathfrak{L}\left(A\to B\right)$
is a real vector space consisting of all the linear maps that take
elements in $\mathrm{Herm}\left(A_{0}\to A_{1}\right)$ to elements
in $\mathrm{Herm}\left(B_{0}\to B_{1}\right)$. In other terms, maps
in $\mathrm{Herm}\left(A\to B\right)$ take Hermitian-preserving maps
to Hermitian-preserving maps. Linear maps in $\mathfrak{L}\left(A\to B\right)$
and $\mathrm{Herm}\left(A\to B\right)$ will be called \emph{supermaps},
and will be denoted by capital Greek letters $\Theta$, $\Upsilon$,
$\Omega$, etc. The identity supermap in $\mathfrak{L}\left(A\to A\right)$
will be denoted by $\mathbbm{1}_{A}$. 

We will use square brackets to denote the action of a supermap $\Theta_{A\to B}\in\mathfrak{L}\left(A\to B\right)$
on a linear map $\mathcal{N}_{A}\in\mathfrak{L}\left(A_{0}\to A_{1}\right)$.
For example, $\Theta_{A\to B}\left[\mathcal{N}_{A}\right]$ is a linear
map in $\mathfrak{L}\left(B_{0}\to B_{1}\right)$ obtained from the
action of the supermap $\Theta$ on the map $\mathcal{N}$. Moreover,
for a simpler notation, the identity supermap will not appear explicitly
in equations; e.g.\ $\Theta_{A\to B}\left[\mathcal{N}_{RA}\right]$
will mean $\left(\mathbbm{1}_{R}\otimes\Theta_{A\to B}\right)\left[\mathcal{N}_{RA}\right]$.
Instead, the action of linear map (e.g.\ quantum channel) $\mathcal{N}_{A}\in\mathfrak{L}\left(A_{0}\to A_{1}\right)$
on a matrix $\rho\in\mathfrak{B}\left(A_{0}\right)$ is written with
round brackets, i.e.\ $\mathcal{N}_{A}\left(\rho_{A_{0}}\right)\in\mathfrak{B}\left(A_{1}\right)$.

Finally, we adopt the following convention concerning partial traces:
when a system is missing, we take the partial trace over the missing
system. This applies to matrices as well as to maps. For example,
if $M_{AB}$ is a matrix on $A_{0}A_{1}B_{0}B_{1}$, $M_{AB_{0}}$
denotes the partial trace on the missing system $B_{1}$: $M_{AB_{0}}:=\mathrm{Tr}_{B_{1}}\left[M_{AB}\right]$.

\subsection{Supermaps and superchannels}

The space $\mathfrak{L}\left(A_{0}\to A_{1}\right)$ is equipped with
an inner product given by
\begin{equation}
\left\langle \mathcal{N}_{A},\mathcal{M}_{A}\right\rangle :=\sum_{i,j}\left\langle \mathcal{N}_{A}\left(\left|i\right\rangle \left\langle j\right|_{A_{0}}\right),\mathcal{M}_{A}\left(\left|i\right\rangle \left\langle j\right|_{A_{0}}\right)\right\rangle _{\mathrm{HS}},\label{inm}
\end{equation}
where $\left\langle X,Y\right\rangle _{\mathrm{HS}}:=\mathrm{tr}\left[X^{\dagger}Y\right]$
is the Hilbert-Schmidt inner product between matrices $X,Y\in\mathfrak{B}\left(A_{1}\right)$.
The inner product above can be expressed in terms of the Choi matrices
of $\mathcal{N}$ and $\mathcal{M}$. Denote by $J_{A}^{\mathcal{N}}:=\mathcal{N}_{\widetilde{A}_{0}\to A_{1}}\left(\phi_{A_{0}\widetilde{A}_{0}}^{+}\right)$
the Choi matrix of $\mathcal{N}_{A}$, where $\phi_{A_{0}\widetilde{A}_{0}}^{+}:=\left|\phi^{+}\right\rangle \left\langle \phi^{+}\right|_{A_{0}\widetilde{A}_{0}}$
and $\left|\phi^{+}\right\rangle _{A_{0}\widetilde{A}_{0}}=\sum_{i}\left|ii\right\rangle _{A_{0}\widetilde{A}_{0}}$
is the unnormalized maximally entangled state. With this notation,
the inner product between $\mathcal{N}_{A}$ and $\mathcal{M}_{A}$
can be expressed as \citep{Gour2018}
\[
\left\langle \mathcal{N}_{A},\mathcal{M}_{A}\right\rangle =\left\langle J_{A}^{\mathcal{N}},J_{A}^{\mathcal{M}}\right\rangle _{\mathrm{HS}}=\mathrm{tr}\left[\left(J_{A}^{\mathcal{N}}\right)^{\dagger}J_{A}^{\mathcal{M}}\right].
\]
The canonical orthonormal basis (relative to the above inner product)
is given by $\left\{ \mathcal{E}_{A}^{ijk\ell}\right\} $, where 
\[
\mathcal{E}_{A}^{ijk\ell}\left(\rho_{A_{0}}\right):=\left\langle i\middle|\rho_{A_{0}}\middle|j\right\rangle \left|k\right\rangle \left\langle \ell\right|_{A_{1}}\quad\forall\rho\in\mathfrak{B}\left(A_{0}\right).
\]

The space $\mathfrak{L}\left(A\to B\right)$ with $A=\left(A_{0},A_{1}\right)$
and $B=\left(B_{0},B_{1}\right)$ is also equipped with the following
inner product: given $\Theta,\Omega\in\mathfrak{L}(A\to B)$ 
\[
\left\langle \Theta_{A\to B},\Omega_{A\to B}\right\rangle :=\sum_{i,j,k,\ell}\left\langle \Theta_{A\to B}\left[\mathcal{E}_{A}^{ijk\ell}\right],\Omega_{A\to B}\left[\mathcal{E}_{A}^{ijk\ell}\right]\right\rangle ,
\]
where the inner product on the right-hand side is the inner product
between maps as defined in Eq.~\eqref{inm}. Similarly to the inner
product between maps, the inner product between supermaps can also
be expressed in terms of Choi matrices. We define the \emph{Choi matrix}
of a supermap $\Theta\in\mathfrak{L}\left(A\to B\right)$ to be \citep{Gour2018}
\[
\mathbf{J}_{AB}^{\Theta}:=\sum_{i,j,k,\ell}J_{A}^{\mathcal{E}^{ijk\ell}}\otimes J_{B}^{\Theta\left[\mathcal{E}_{A}^{ijk\ell}\right]}.
\]
Then, with this notation, the inner product between two supermaps
$\Theta$ and $\Omega$ can be expressed as 
\[
\left\langle \Theta_{A\to B},\Omega_{A\to B}\right\rangle =\left\langle \mathbf{J}_{AB}^{\Theta},\mathbf{J}_{AB}^{\Omega}\right\rangle _{\mathrm{HS}}=\mathrm{tr}\left[\left(\mathbf{J}_{AB}^{\Theta}\right)^{\dagger}\mathbf{J}_{AB}^{\Omega}\right].
\]

The Choi matrix of a supermap $\Theta\in\mathfrak{L}\left(A\to B\right)$
can also be expressed in other three alternative ways \citep{Gour2018}.
First, from its definition, $\mathbf{J}_{AB}^{\Theta}$ can be expressed
as the Choi matrix of the map 
\[
\mathcal{P}_{AB}^{\Theta}:=\Theta_{\widetilde{A}\to B}\left[\Phi_{A\widetilde{A}}^{+}\right],
\]
where the map $\Phi_{A\widetilde{A}}^{+}$ is defined as 
\[
\Phi_{A\widetilde{A}}^{+}:=\sum_{i,j,k,\ell}\mathcal{E}_{A}^{ijk\ell}\otimes\mathcal{E}_{\widetilde{A}}^{ijk\ell}.
\]
A simple calculation shows that $\Phi_{A\widetilde{A}}^{+}$ is completely
positive, and acts on $\rho\in\mathfrak{B}\left(A_{0}\widetilde{A}_{0}\right)$
as 
\[
\Phi_{A\widetilde{A}}^{+}\left(\rho_{A_{0}\widetilde{A}_{0}}\right)=\mathrm{tr}\left[\rho_{A_{0}\widetilde{A}_{0}}\phi_{A_{0}\widetilde{A}_{0}}^{+}\right]\phi_{A_{1}\widetilde{A}_{1}}^{+}.
\]
In other terms, the CP map $\Phi_{A\widetilde{A}}^{+}$ can be viewed
as a generalization of the (unnormalized) maximally entangled state
$\phi_{A_{0}\widetilde{A}_{0}}^{+}$.

A supermap $\Theta\in\mathfrak{L}\left(A\to B\right)$ can also be
characterized by its action on Choi matrices. One can define a linear
map $\mathcal{R}^{\Theta}:\mathfrak{B}\left(A\right)\to\mathfrak{B}\left(B\right)$
as 
\[
\mathcal{R}_{A\to B}^{\Theta}\left(\rho_{A}\right):=\mathrm{tr}_{A}\left[\mathbf{J}_{AB}^{\Theta}\left(\rho_{A}^{T}\otimes I_{B}\right)\right]\quad\forall\rho\in\mathfrak{B}\left(A\right).
\]
With this definition, $\mathbf{J}_{AB}^{\Theta}$ can be viewed as
the Choi matrix of $\mathcal{R}_{A\to B}^{\Theta}$. Note that although
$\mathcal{P}_{AB}^{\Theta}$ and $\mathcal{R}_{A\to B}^{\Theta}$
have the same Choi matrix $\mathbf{J}_{AB}^{\Theta}$, $\mathcal{P}_{AB}^{\Theta}$
takes systems $A_{0}B_{0}$ to $A_{1}B_{1}$, whereas the map $\mathcal{R}^{\Theta}$
takes system $A=\left(A_{0},A_{1}\right)$ to system $B=\left(B_{0},B_{1}\right)$.
This brings us to the last representation of a supermap in terms of
a linear map $\mathcal{Q}^{\Theta}:\mathfrak{B}\left(A_{1}B_{0}\right)\to\mathfrak{B}\left(A_{0}B_{1}\right)$,
which is defined as the map satisfying 
\begin{equation}
\mathbf{J}_{AB}^{\Theta}:=\mathcal{Q}_{\widetilde{A}_{1}\widetilde{B}_{0}\to A_{0}B_{1}}^{\Theta}\left(\phi_{A_{1}\widetilde{A}_{1}}^{+}\otimes\phi_{B_{0}\widetilde{B}_{0}}^{+}\right),\label{eq:def Q}
\end{equation}
or as $\mathcal{Q}^{\Theta}:=\mathbbm{1}_{A}\otimes\Theta_{A\rightarrow B}\left[\mathcal{S}_{A}\right]$,
where $\mathcal{S}_{A}$ is the swap from $A_{1}$ to $A_{0}$. All
these three representations of a supermap, $\mathcal{P}^{\Theta}$,
$\mathcal{Q}^{\Theta}$, and $\mathcal{R}^{\Theta}$, play a useful
role in the study of quantum resource theories, as shown in Ref.~\citep{Gour-Scandolo}
in the case of the entanglement of bipartite channels.

A \emph{superchannel} is a supermap $\Theta_{A\to B}\in\mathfrak{L}\left(A\to B\right)$
that takes quantum channels to quantum channels even when tensored
with the identity supermap \citep{Chiribella2008,Switch,Hierarchy-combs,Perinotti1,Perinotti2,Gour2018,Supermeasurements}.
More precisely, $\Theta_{A\to B}\in\mathfrak{L}\left(A\to B\right)$
is called a superchannel if it satisfies the following two conditions: 
\begin{enumerate}
\item For any trace-preserving map $\mathcal{N}_{A}\in\mathfrak{L}\left(A_{0}\to A_{1}\right)$,
the map $\Theta_{A\to B}\left[\mathcal{N}_{A}\right]$ is a trace-preserving
map in $\mathfrak{L}\left(B_{0}\to B_{1}\right)$.
\item For any system $R=\left(R_{0},R_{1}\right)$ and any bipartite CP
map $\mathcal{N}_{RA}\in\mathrm{CP}\left(R_{0}A_{0}\to R_{1}A_{1}\right)$,
the map $\Theta_{A\to B}\left[\mathcal{N}_{RA}\right]$ is also CP.
\end{enumerate}
We will also say that a supermap $\Theta_{A\to B}\in\mathfrak{L}\left(A\to B\right)$,
is \emph{positive} if it takes CP maps to CP maps, and \emph{completely
positive} (CP), if it satisfies the second condition above \citep{Chiribella2008,Gour2018}.
Therefore, a superchannel is a CP supermap that takes trace-preserving
maps to trace-preserving maps \citep{Gour2018,Supermeasurements}.
We will denote the set of superchannels from $A$ to $B$ by $\mathfrak{S}\left(A\rightarrow B\right)$.
Note that $\mathfrak{S}\left(A\rightarrow B\right)\subset\mathfrak{L}\left(A\rightarrow B\right)$.

The above definition is axiomatic and minimalist, in the sense that
any physical evolution (or simulation) of a quantum channel must satisfy
these two basic conditions. The third part of the following theorem
shows that these two conditions are sufficient to ensure that superchannels
are indeed \emph{physical} processes.
\begin{thm}[\citep{Chiribella2008,Gour2018}]
\label{premain}Let $\Theta\in\mathfrak{L}\left(A\to B\right)$.
The following are equivalent.
\begin{enumerate}
\item $\Theta$ is a superchannel.
\item The Choi matrix $\mathbf{J}_{AB}^{\Theta}\geq0$ of $\Theta$ has
marginals 
\begin{equation}
\mathbf{J}_{A_{1}B_{0}}^{\Theta}=I_{A_{1}B_{0}},\qquad\mathbf{J}_{AB_{0}}^{\Theta}=\mathbf{J}_{A_{0}B_{0}}^{\Theta}\otimes u_{A_{1}},\label{marginals}
\end{equation}
where $u_{A_{1}}$ is the maximally mixed state (i.e.\ the uniform
state) on system $A_{1}$.
\item There exists a Hilbert space $E$, with $\left|E\right|\leq\left|A_{0}B_{0}\right|$,
and two CPTP maps $\mathcal{F}\in\mathrm{CPTP}\left(B_{0}\to EA_{0}\right)$
and $\mathcal{E}\in\mathrm{CPTP}\left(EA_{1}\to B_{1}\right)$ such
that for all $\mathcal{N}_{A}\in\mathfrak{L}\left(A_{0}\to A_{1}\right)$
\begin{equation}
\Theta\left[\mathcal{N}_{A}\right]=\mathcal{E}_{EA_{1}\to B_{1}}\circ\mathcal{N}_{A_{0}\to A_{1}}\circ\mathcal{F}_{B_{0}\to EA_{0}}\label{realization0}
\end{equation}
(see Fig.~\ref{superchannel}).
\begin{figure}
\begin{centering}
\includegraphics[width=1\columnwidth]{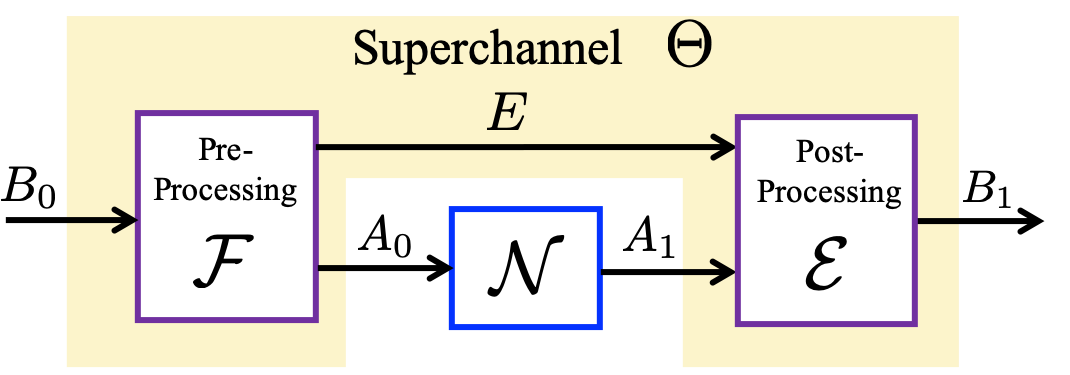}
\par\end{centering}
\caption{\label{superchannel}Realization of a superchannel in terms of a pre-processing
channel $\mathcal{F}$ and a post-processing channel $\mathcal{E}$.}
\end{figure}
 Furthermore, $\mathcal{Q}_{A_{1}B_{0}\to A_{0}B_{1}}^{\Theta}=\mathcal{E}_{EA_{1}\to B_{1}}\circ\mathcal{F}_{B_{0}\to EA_{0}}\in\mathrm{CPTP}\left(A_{1}B_{0}\to A_{0}B_{1}\right)$,
and $\mathcal{F}$ can be taken to be an isometry.
\item For every $\mathcal{N}\in\mathrm{CPTP}\left(A_{0}\to A_{1}\right)$,
the matrix $\mathcal{R}_{A\to B}^{\Theta}\left(J_{A}^{\mathcal{N}}\right)$
is a Choi matrix of a quantum channel. That is, 
\[
\mathcal{R}_{A\to B}^{\Theta}\left(J_{A}^{\mathcal{N}}\right)\geq0\text{ and }\mathrm{tr}_{B_{1}}\left[\mathcal{R}_{A\to B}^{\Theta}\left(J_{A}^{\mathcal{N}}\right)\right]=I_{B_{0}}.
\]
\end{enumerate}
\end{thm}

In general, the realization of a superchannel as given in Fig.~\ref{superchannel}
is not unique. This is due to the presence of a memory system, described
in Fig.~\ref{superchannel} with the letter $E$. To see why, consider
an isometry channel $\mathcal{V}_{E\to E'}$ defined for all $\rho\in\mathfrak{B}\left(E\right)$
by $\mathcal{V}_{E\to E'}\left(\rho\right)=V\rho_{E}V^{\dagger}$,
where $V:E\to E'$ is an isometry matrix satisfying $V^{\dagger}V=I_{E}$.
Then, this isometry channel has many left inverses given by 
\[
\mathcal{V}_{E'\to E}^{-1}\left(\sigma_{E'}\right)=V^{\dagger}\sigma_{E'}V+\mathrm{tr}\left[\left(I_{E'}-VV^{\dagger}\right)\sigma_{E'}\right]\tau_{E},
\]
where $\tau\in\mathfrak{D}\left(E\right)$ is an arbitrary fixed density
matrix. We can easily check that $\mathcal{V}^{-1}\circ\mathcal{V}=\mathsf{id}$.
In Fig.~\ref{realization} we use this map to show that the realization
of a superchannel in terms of pre- and post-processing is \emph{not}
unique.
\begin{figure}
\begin{centering}
\includegraphics[width=1\columnwidth]{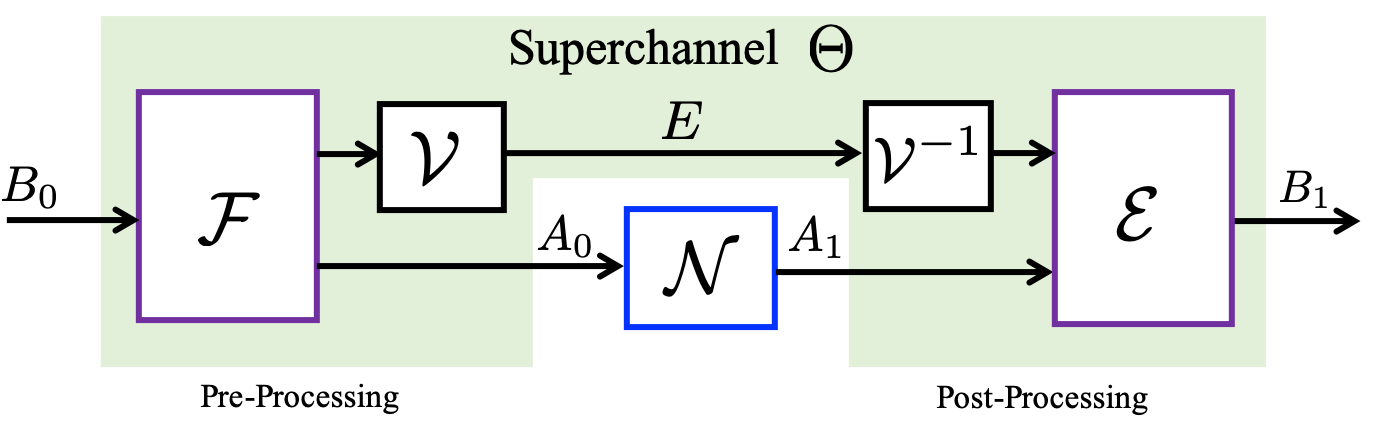}
\par\end{centering}
\caption{\label{realization}The realization of a superchannel is not unique.
The map $\mathcal{V}$ can be any linear map (not even a channel)
for which there exists another linear map $\mathcal{V}^{-1}$ such
that $\mathcal{V}^{-1}\circ\mathcal{V}=\mathsf{id}$. For example,
if $\mathcal{V}$ is an isometry channel that is not unitary, there
are many channels $\mathcal{V}^{-1}$ that satisfy $\mathcal{V}^{-1}\circ\mathcal{V}=\mathsf{id}$.
Note that one can even take $\mathcal{V}=\mathcal{V}^{-1}=\mathcal{T}$
to be the transpose map, in which case the resulting pre- and post-processing
are not even necessarily CP!}
\end{figure}
 Moreover, there is another way in which the realization of a superchannel
can be non-unique, namely by appending a state in the pre-processing,
and then discarding it in the post-processing. To see how this works,
let $\mathcal{F}_{B_{0}\to EA_{0}}$ and $\mathcal{E}_{EA_{1}\to B_{1}}$
be the pre-processing and the post-processing in a realization of
a superchannel $\Theta\in\mathfrak{S}\left(A\to B\right)$, respectively.
Now consider the new pre-processing $\mathcal{F}'_{B_{0}\to E'EA_{0}}:=\rho_{E'}\otimes\mathcal{F}_{B_{0}\to EA_{0}}$,
where $\rho\in\mathfrak{D}\left(E'\right)$, and the new post-processing
$\mathcal{E}_{A_{1}EE'\to B_{1}}:=\mathrm{tr}_{E'}\otimes\mathcal{E}_{EA_{1}\to B_{1}}$.
It is straightforward to check that $\mathcal{F}'$ and $\mathcal{E}'$
realize exactly the same superchannel $\Theta$, as $\mathcal{F}$
and $\mathcal{E}$.

Although the realization of a superchannel is not unique, if we restrict
the dimension of system $E$ to be the smallest possible, and the
map $\mathcal{F}$ to be an isometry, we can obtain a new uniqueness
result, expressed by the following theorem, which subsumes some of
the results in Ref.~\citep{Bisio-minimal}.
\begin{thm}[Uniqueness]
\label{thm:uniqueness}Let $\Theta\in\mathfrak{S}\left(A\to B\right)$
be a superchannel, and let $r:=\mathrm{Rank}\left(\mathbf{J}_{A_{0}B_{0}}^{\Theta}\right)$.
Then, there exists a system $E$ with $\left|E\right|=r$, an isometry
$\mathcal{F}\in\mathrm{CPTP}\left(B_{0}\to EA_{0}\right)$ and a channel
$\mathcal{E}\in\mathrm{CPTP}\left(EA_{1}\to B_{1}\right)$ such that
$\Theta$ can be realized as in Eq.~\eqref{realization0}. Furthermore,
if there exists a system $E'$ such that $\left|E'\right|\leq r$,
an isometry $\mathcal{F}'\in\mathrm{CPTP}\left(B_{0}\to E'A_{0}\right)$,
and a channel $\mathcal{E}'\in\mathrm{CPTP}\left(E'A_{1}\to B_{1}\right)$
such that $\Theta$ can be realized as in Eq.~\eqref{realization0}
with $\mathcal{F}'$ and $\mathcal{E}'$ replacing $\mathcal{E}$
and $\mathcal{F}$, then $\left|E'\right|=\left|E\right|$, and there
exists a unitary channel $\mathcal{U}\in\mathrm{CPTP}\left(E\to E'\right)$
such that 
\[
\mathcal{E}'_{E'A_{1}\to B_{1}}=\mathcal{E}_{EA_{1}\to B_{1}}\circ\mathcal{U}_{E'\to E}^{-1}
\]
and
\[
\mathcal{F}'_{B_{0}\to E'A_{0}}=\mathcal{U}_{E\to E'}\circ\mathcal{F}_{B_{0}\to EA_{0}}.
\]
\end{thm}

\begin{proof}
The first part of the theorem follows from the proof of Theorem~\ref{premain}
as given in Ref.~\citep{Gour2018}, in which system $E$ was chosen
to be the purifying system of $\mathbf{J}_{A_{0}B_{0}}^{\Theta}$
(see also Ref.~\citep{Bisio-minimal}). Thus $\left|E\right|$ can
always be taken to have dimension $\left|E\right|=r$. We only need
to prove the uniqueness part.

First note that by Theorem~\ref{premain} we have that 
\begin{align*}
\mathcal{Q}_{A_{0}B_{1}\to A_{1}B_{0}}^{\Theta} & =\mathcal{E}'_{E'A_{1}\to B_{1}}\circ\mathcal{F}'_{B_{0}\to E'A_{0}}\\
 & =\mathcal{E}_{EA_{1}\to B_{1}}\circ\mathcal{F}_{B_{0}\to EA_{0}},
\end{align*}
whose Choi matrix is $\mathbf{J}_{AB}^{\Theta}$. Therefore, recalling
Eq.~\eqref{eq:def Q}, the marginal$\mathbf{J}_{A_{0}B_{0}}^{\Theta}$
can be expressed as 
\[
\mathbf{J}_{A_{0}B_{0}}^{\Theta}=\left|A_{1}\right|\mathcal{F}'_{B_{0}\to A_{0}}\left(\phi_{B_{0}\widetilde{B}_{0}}^{+}\right)=\left|A_{1}\right|\mathcal{F}_{B_{0}\to A_{0}}\left(\phi_{B_{0}\widetilde{B}_{0}}^{+}\right).
\]
Now, observe that $\left|A_{1}\right|\mathcal{F}'_{\widetilde{B}_{0}\to E'A_{0}}\left(\phi_{B_{0}\widetilde{B}_{0}}^{+}\right)$
is a purification of $\mathbf{J}_{A_{0}B_{0}}^{\Theta}$ since by
assumption $\mathcal{F}'_{\widetilde{B}_{0}\to E'A_{0}}$ is an isometry
(whence $\left|A_{1}\right|\mathcal{F}'_{\widetilde{B}_{0}\to E'A_{0}}\left(\phi_{B_{0}\widetilde{B}_{0}}^{+}\right)$
is a pure state). Therefore, $\left|E'\right|\geq r$ so that $\left|E'\right|=r=\left|E\right|$.
Moreover, since $\left|A_{1}\right|\mathcal{F}'_{\widetilde{B}_{0}\to E'A_{0}}\left(\phi_{B_{0}\widetilde{B}_{0}}^{+}\right)$
and $\left|A_{1}\right|\mathcal{F}_{\widetilde{B}_{0}\to EA_{0}}\left(\phi_{B_{0}\widetilde{B}_{0}}^{+}\right)$
are two purifications of $\mathbf{J}_{A_{0}B_{0}}^{\Theta}$, they
must be related by a unitary $\mathcal{U}_{E\to E'}$, so $\mathcal{F}'_{B_{0}\to E'A_{0}}=\mathcal{U}_{E\to E'}\circ\mathcal{F}_{B_{0}\to EA_{0}}$,
as their Choi matrices are the same.

To conclude the proof, set $\psi_{EA_{0}B_{0}}:=\mathcal{F}_{\widetilde{B}_{0}\to EA_{0}}\left(\phi_{B_{0}\widetilde{B}_{0}}^{+}\right)$.
Recalling that $\mathbf{J}_{AB}^{\Theta}$ is the Choi matrix of $\mathcal{Q}^{\Theta}$,
we get
\begin{align}
\mathbf{J}_{AB}^{\Theta} & =\mathcal{E}_{E\widetilde{A}_{1}\to B_{1}}\left(\psi_{EA_{0}B_{0}}\otimes\phi_{A_{1}\widetilde{A}_{1}}^{+}\right)\nonumber \\
 & =\mathcal{E}'_{E'\widetilde{A}_{1}\to B_{1}}\circ\mathcal{U}_{E\to E'}\left(\psi_{EA_{0}B_{0}}\otimes\phi_{A_{1}\widetilde{A}_{1}}^{+}\right).\label{eq:uniqueness equality}
\end{align}
Let system $\widetilde{E}$ be the support of $\psi_{A_{0}B_{0}}$
(i.e.\ it is the Hilbert space spanned by the eigenvectors of $\psi_{A_{0}B_{0}}$
that correspond to non-zero eigenvalues). Hence, $\left|\widetilde{E}\right|=\left|E\right|=r$.
Denoting the restriction of $\psi_{EA_{0}B_{0}}$ to the space $E\widetilde{E}$
by $\psi_{E\widetilde{E}}$, by Eq.~\eqref{eq:uniqueness equality}
we have that
\begin{align}
 & \mathcal{E}_{E\widetilde{A}_{1}\to B_{1}}\left(\psi_{E\widetilde{E}}\otimes\phi_{A_{1}\widetilde{A}_{1}}^{+}\right)\nonumber \\
 & =\mathcal{E}'_{E'\widetilde{A}_{1}\to B_{1}}\circ\mathcal{U}_{E\to E'}\left(\psi_{E\widetilde{E}}\otimes\phi_{A_{1}\widetilde{A}_{1}}^{+}\right)\label{eq:uniqueness equality2}
\end{align}
By definition, the marginal $\psi_{\widetilde{E}}$ is invertible,
and we have $\psi_{E\widetilde{E}}=\left(I_{E}\otimes\sqrt{\psi_{\widetilde{E}}}U_{\widetilde{E}}\right)\phi_{E\widetilde{E}}^{+}\left(I_{E}\otimes U_{\widetilde{E}}^{\dagger}\sqrt{\psi_{\widetilde{E}}}\right)$,
where $U_{\widetilde{E}}$ is some unitary. Hence, by (Hermite-) conjugating
both sides of Eq.~\eqref{eq:uniqueness equality2} above by $U_{\widetilde{E}}^{\dagger}\psi_{\widetilde{E}}^{-\frac{1}{2}}$,
we get that the Choi matrix of $\mathcal{E}_{E\widetilde{A}_{1}\to B_{1}}$
equals the Choi matrix of $\mathcal{E}'_{E'\widetilde{A}_{1}\to B_{1}}\circ\mathcal{U}_{E\to E'}$.
Consequently we conclude that the channels must be the same.
\end{proof}

\subsection{Measurements on quantum channels}

A quantum instrument is a collection of CP maps $\left\{ \mathcal{E}_{x}\right\} $
such that their sum $\sum_{x}\mathcal{E}_{x}$ is a CPTP map. Note
that each $\mathcal{E}^{x}$ is trace non-increasing, and that every
CP map that is trace non-increasing can be completed to a full quantum
instrument. Quantum instruments are used to characterize the most
general measurements that can be performed on a physical system, including,
as special cases, projective von Neumann measurements, POVMs, and
generalized measurements. Therefore, we discuss the generalization
of a quantum instrument to a collection of objects that act on quantum
channels. We call this generalization a\emph{ superinstrument} \citep{Supermeasurements}.

A superinstrument is a collection of supermaps $\left\{ \Theta_{x}\right\} $,
where each $\Theta_{x}\in\mathfrak{L}\left(A\to B\right)$ is CP (i.e.\ $\mathbf{J}_{AB}^{\Theta_{x}}\geq0$),
and the sum $\sum_{x}\Theta_{x}$ is a superchannel. Similar to the
state domain, every $\Theta_{x}$ maps quantum channels to CP trace
non-increasing maps. However, in the channel domain not every supermap
$\Theta\in\mathfrak{L}\left(A\to B\right)$ with a positive semi-definite
Choi matrix, and that takes channels to CP trace non-increasing maps,
can be completed to a superchannel. In Ref.~\citep{Supermeasurements}
a counterexample was given, and it was also shown that a CP supermap
$\Theta\in\mathfrak{L}\left(A\to B\right)$ can be completed to a
superchannel (i.e.\ there exists a CP supermap $\Omega\in\mathfrak{L}\left(A\to B\right)$
such that $\Theta+\Omega$ is a superchannel) if and only if for any
system $R$, the supermap $\mathbbm{1}_{R}\otimes\Theta$ takes quantum
channels to CP trace non-increasing maps. In Ref.~\citep{Supermeasurements}
it was shown that this phenomenon is associated with the existence
of signaling bipartite channels.

While the above discussion is subtle, it demonstrates (see details
in Ref.~\citep{Supermeasurements}) that every element $\Theta_{x}$
of a superinstrument $\left\{ \Theta_{x}\right\} $ satisfies $\mathrm{tr}\left[\mathbf{J}_{AB_{0}}^{\Theta_{x}}\alpha_{AB_{0}}\right]\leq1$
for every $\alpha_{AB_{0}}\geq0$ such that $\alpha_{A_{0}B_{0}}=I_{A_{0}}\otimes\rho_{B_{0}}$,
where $\rho\in\mathfrak{D}\left(B_{0}\right)$. Moreover, every superinstrument
can be realized as in Fig.~\ref{instrument}, with an isometry pre-processing
and a quantum instrument as the post-processing \citep{Chiribella2008,Supermeasurements}.
\begin{figure}
\begin{centering}
\includegraphics[width=1\columnwidth]{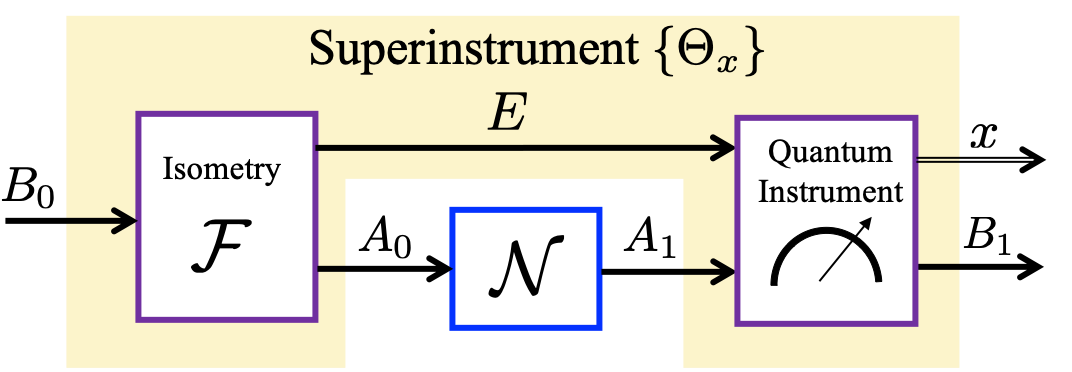}
\par\end{centering}
\caption{\label{instrument}The realization of a superinstrument. The map $\mathcal{F}$
can be taken to be an isometry and the post-processing is a quantum
instrument.}
\end{figure}
 Like quantum instruments, any superinstrument $\left\{ \Theta_{x}\right\} $
in $\mathfrak{L}\left(A\to B\right)$ can be viewed as a superchannel
$\Theta\in\mathfrak{L}\left(A\to BX\right)$, where system $X=\left(X_{0},X_{1}\right)$
has trivial input dimension $\left|X_{0}\right|=1$, and the output
system $X_{1}$ is classical. Hence, a superinstrument can be expressed
as 
\[
\Theta_{A\to BX}=\sum_{x}\Theta_{A\to B}^{x}\otimes\left|x\right\rangle \left\langle x\right|_{X},
\]
where $X\equiv X_{1}$. This characterization of a superinstrument
is particularly useful in the context of quantum resource theories,
since the above relation demonstrates that the set of free superinstruments
can be viewed as a subset of the set of free superchannels.

\subsection{Quantum combs}

Quantum combs are multipartite channels with a well-defined causal
structure (see Fig.~\ref{comb}(a)) \citep{Gutoski,Circuit-architecture,Hierarchy-combs,Gutoski2,Chiribella2016,Gutoski3}.
\begin{figure}
\begin{centering}
\includegraphics[width=1\columnwidth]{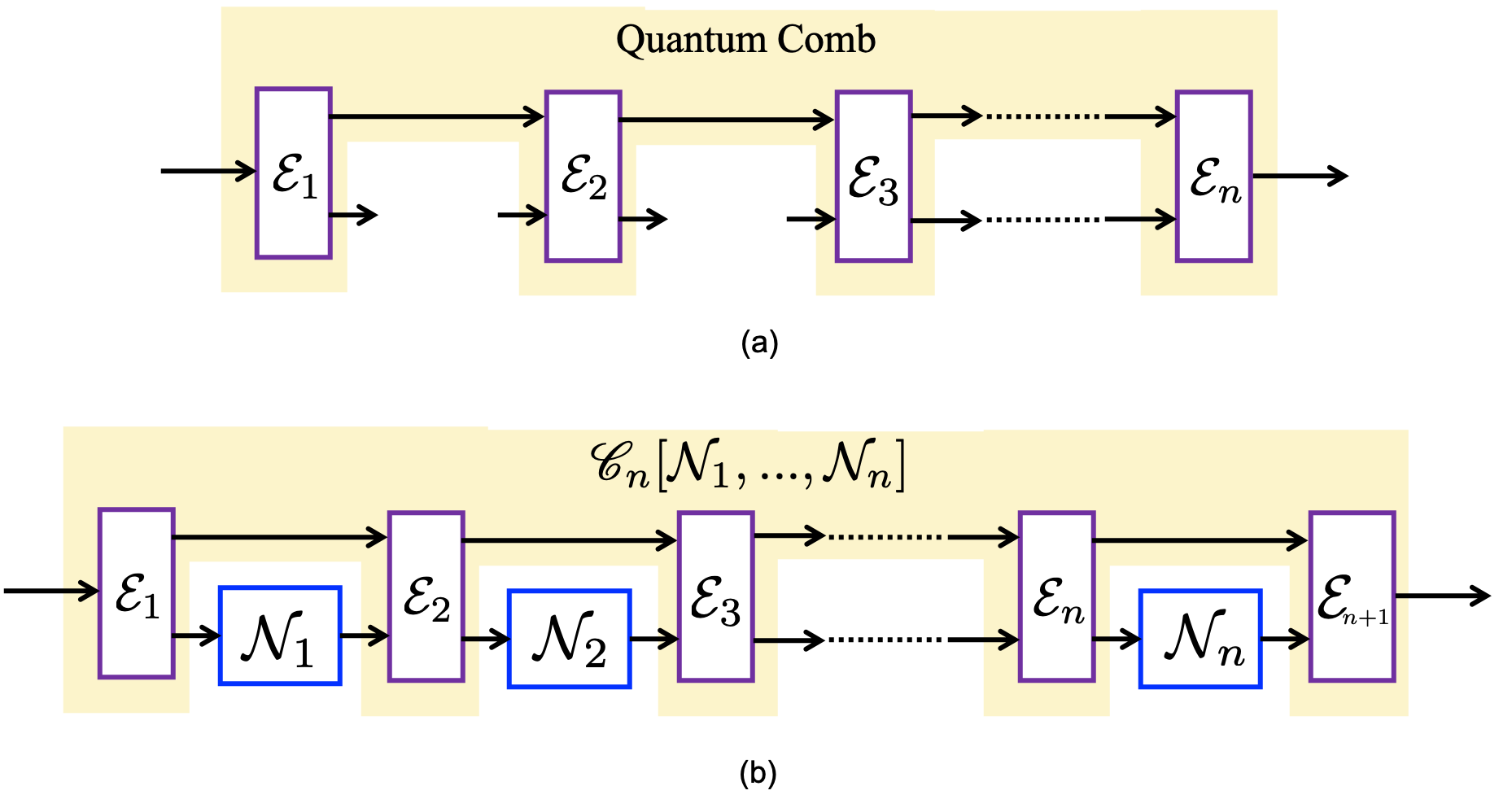}
\par\end{centering}
\caption{\label{comb}(a) A quantum comb that can be realized with $n$ channels.
(b) The action of $\mathscr{C}_{n}$ on $n$ channels $\mathcal{N}_{1},\dots,\mathcal{N}_{n}$.
Note that the input channels are\emph{ causally ordered} in the slots
of the comb from left to right, i.e.\ from $\mathcal{N}_{1}$ to
$\mathcal{N}_{n}$.}
\end{figure}
 They generalize the notion of superchannels to objects that take
several channels as input, and output a channel (see Refs.~\citep{Circuit-architecture,Hierarchy-combs}
for more details, and a for a further generalization where the input
and the output of combs are combs themselves). A comb acting on $n$
channels is depicted in Fig.~\ref{comb}(b). We will denote a comb
with $n$ channel-slots as input by $\mathscr{C}_{n}$, and its action
on $n$ channels by $\mathscr{C}_{n}\left[\mathcal{N}_{1},\dots,\mathcal{N}_{n}\right]$.
The causal relation between the different slots ensures that each
such comb can be realized with $n+1$ channels $\mathcal{E}_{1},\dots,\mathcal{E}_{n+1}$
as in Fig.~\ref{comb}(b). We therefore associate a quantum channel
\[
\mathcal{Q}^{\mathscr{C}_{n}}:=\mathcal{E}_{n+1}\circ\mathcal{E}_{n}\circ\dots\circ\mathcal{E}_{1}
\]
with every comb. Note that the quantum channel $\mathcal{Q}^{\mathscr{C}_{n}}$
has a causal structure in the sense that the input to $\mathcal{E}_{k}$
cannot affect the output of $\mathcal{E}_{k-1}$ for any $k=2,\dots,n+1$.
The Choi matrix of the comb is defined as the Choi matrix of $\mathcal{Q}^{\mathscr{C}_{n}}$.
Owing to the causal structure of $\mathcal{Q}^{\mathscr{C}_{n}}$,
the marginals of the Choi matrix of $\mathscr{C}_{n}$ satisfy similar
relations to Eq.~\eqref{marginals} (see Refs.~\citep{Circuit-architecture,Hierarchy-combs}
for more details).

Note that there are other ways to manipulate multiple quantum channels
where we do not require any causal structure on the different channel-slots
\citep{Switch,Process-matrix,Giarmatzi}, but we will not use them
in our analysis.

\subsection{Quantum resource theories}

For every pair of physical systems $A$ and $B$, consider a subset
of CPTP maps $\mathfrak{F}\left(A\to B\right)\subset\mathrm{CPTP}\left(A\to B\right)$.
$\mathfrak{F}$ identifies a \emph{quantum resource theory} (QRT)
if the following two conditions hold \citep{Gour-review}:
\begin{enumerate}
\item For every physical system $A$, the set $\mathfrak{F}\left(A\to A\right)$
contains the identity map $\mathsf{id}_{A}$.
\item For any three systems $A$, $B$, $C$, if $\mathcal{M}\in\mathfrak{F}\left(A\to B\right)$
and $\mathcal{N}\in\mathfrak{F}\left(B\to C\right)$, then $\mathcal{N}\circ\mathcal{M}\in\mathfrak{F}\left(A\to C\right)$.
\end{enumerate}
The elements in each set $\mathfrak{F}\left(A\to B\right)$ are called
\emph{free operations}. The set $\mathfrak{F}\left(A\right):=\mathfrak{F}\left(1\to A\right)$,
where the $1$ stands for the trivial (i.e.\ 1-dimensional) system,
will be used to denote the set of free states. 

In any QRT we can consider either static or dynamical inter-conversions.
In a static inter-conversion we look for conditions under which a
conversion from one resource state (i.e.\ not in $\mathfrak{F}\left(A\right)$)
to another is possible by free operations. In a dynamical inter-conversion
we are interested in the conditions under which a conversion from
one resource channel (i.e.\ not in $\mathfrak{F}\left(A\to B\right)$)
to another is possible with \emph{free superchannels}. Clearly, static
inter-conversions can be viewed as a special type of dynamical ones. 

In this article we will consider QRTs that admit a tensor product
structure. That is, the set of free operations $\mathfrak{F}$ satisfies
the following additional conditions:
\begin{enumerate}
\item[3.] Free operations are ``completely free'': for any three physical
systems $A$, $B$, and $C$, if $\mathcal{M}\in\mathfrak{F}\left(A\to B\right)$
then $\mathsf{id}_{C}\otimes\mathcal{M}\in\mathfrak{F}\left(CA\to CB\right)$. 
\item[4.] Discarding a system (i.e.\ the trace) is a free operation: for every
system $A$, the set $\mathfrak{F}\left(A\to1\right)$ is not empty.
\end{enumerate}
The above additional conditions are very natural, and satisfied by
almost all QRTs studied in literature \citep{Gour-review}. These
conditions imply that if $\mathcal{M}_{1}$ and $\mathcal{M}_{2}$
are free channels, then also $\mathcal{M}_{1}\otimes\mathcal{M}_{2}$
is free. In addition, they also imply that appending free states is
a free operation; i.e.\ for any given free state $\sigma\in\mathfrak{F}\left(B\right)$,
the CPTP map $\mathcal{F}_{\sigma}\left(\rho\right):=\rho\otimes\sigma$
is a free map, i.e.\ it belongs to $\mathfrak{F}\left(A\to AB\right)$.
This in turn implies that the replacement map $\mathcal{R}_{\sigma}$
is free, where $\mathcal{R}_{\sigma}\left(\rho\right)=\mathrm{tr}\left[\rho\right]\sigma$,
for every density matrix $\rho$, and some fixed free state $\sigma$.
In the following we will also assume that $\mathfrak{F}\left(A\to B\right)$
is topologically closed for all systems $A$ and $B$, as it is natural
to assume that arbitrarily good approximations of free operations
are free as well.

\section{Resource theories of quantum processes\label{sec:Resource-processes}}

In this section we build resource theories of processes, and we present
a new construction of a complete set of monotones for convex resource
theories of processes. We also give the precise definition of several
resource-theoretic protocols.

Similarly to resource theories of quantum states, free superchannels
will be a \emph{subset} of all physical superchannels. If we already
have a QRT of static resources, theorem~\ref{premain} gives us a
\emph{sufficient condition} for free superchannels: a superchannel
is free if both the pre-processing and the post-processing are free
in the underlying resource theory of states, i.e.\ if $\mathcal{F}\in\mathfrak{F}\left(B_{0}\to EA_{0}\right)$
and $\mathcal{E}\in\mathfrak{F}\left(EA_{1}\to B_{1}\right)$. We
call these free superchannels ``freely realizable''. Since we consider
QRTs with a tensor product structure, if a superchannel $\Theta$
is free, then also its map
\begin{equation}
\mathcal{Q}_{A_{1}B_{0}\to A_{0}B_{1}}^{\Theta}:=\mathcal{E}_{EA_{1}\to B_{1}}\circ\mathcal{F}_{B_{0}\to EA_{0}}\label{decompose}
\end{equation}
is free: $\mathcal{Q}_{A_{1}B_{0}\to A_{0}B_{1}}^{\Theta}\in\mathfrak{F}\left(A_{1}B_{0}\to A_{0}B_{1}\right)$.
Recall that the mapping $\Theta\mapsto\mathcal{Q}_{A_{1}B_{0}\to A_{0}B_{1}}^{\Theta}$
is a bijection, so that a free superchannel $\Theta$ corresponds
to a unique free map $\mathcal{Q}_{A_{1}B_{0}\to A_{0}B_{1}}^{\Theta}$.
However, if $\mathcal{Q}_{A_{1}B_{0}\to A_{0}B_{1}}^{\Theta}$ is
a free CPTP map, it does \emph{not} necessarily mean that there exists
a realization of $\Theta$ in terms of free pre- and post-processing:
we only know that their combination is free.

The problem of determining whether a free channel $\mathcal{Q}$ can
be decomposed as in Eq.~\eqref{decompose} with both $\mathcal{E}$
and $\mathcal{F}$ being free can be very hard to solve, even when
the resource theory is relatively simple (that is, even if inclusion
in $\mathfrak{F}$ can be determined with an SDP; e.g.\ in NPT entanglement,
see Ref.~\citep{Dynamical-entanglement}). Therefore, typically,
resource theories of quantum processes can be very hard to handle,
even if the corresponding QRT of states is relatively simple. In Ref.~\citep{Dynamical-entanglement},
we announced that, for NPT dynamical entanglement, if we enlarge the
set of free superchannels to include \emph{all} superchannels for
which $\mathcal{Q}$ is a PPT channel, we obtain a resource theory
of NPT dynamical entanglement that is much more manageable. The price
we pay is that not all such free superchannels may be freely realizable.

In view of this more relaxed definition of free superchannels, let
us focus on the minimal requirements for free superchannels. For any
two systems $A$ and $B$, we denote by $\mathrm{FREE}\left(A\to B\right)$
the set of all free superchannels in $\mathfrak{S}\left(A\to B\right)$.
The minimal requirements the set $\mathrm{FREE}$ must satisfy are
the following (analogous to those satisfied by $\mathfrak{F}$):
\begin{enumerate}
\item $\mathbbm{1}_{A}\in\mathrm{FREE}\left(A\to A\right)$, where $\mathbbm{1}_{A}$
is the identity supermap acting on $\mathfrak{L}\left(A\to A\right)$.
\item If $\Theta_{1}\in\mathrm{FREE}\left(A\to B\right)$ and $\Theta_{2}\in\mathrm{FREE}\left(B\to C\right)$,
then $\Theta_{2}\circ\Theta_{1}\in\mathrm{FREE}\left(A\to C\right)$.
\end{enumerate}
In particular, the second condition also implies that the superchannels
in $\mathrm{FREE}$ are resource non-generating (RNG) \citep{Quantum-resource-2,Gour-review}.
In other words, for every input channel $\mathcal{M}_{A}\in\mathfrak{F}\left(A_{0}\to A_{1}\right)$
and every free superchannel $\Theta\in\mathrm{FREE}\left(A\to B\right)$,
the output channel $\Theta\left[\mathcal{M}_{A}\right]\in\mathfrak{F}\left(B_{0}\to B_{1}\right)$.
Note that we can recover free channels by trivializing the input $A$
of a free superchannel $\Theta_{A\rightarrow B}$, i.e.\ by taking
$A_{0}$ and $A_{1}$ to be 1-dimensional.

Moreover, since we consider QRTs that admit a tensor product structure,
we require free superchannels to be ``completely free'': for any
three physical systems $A=\left(A_{0},A_{1}\right)$, $B=\left(B_{0},B_{1}\right)$,
and $R=\left(R_{0},R_{1}\right)$, if $\Theta\in\mathrm{FREE}\left(A\to B\right)$,
then $\mathbbm{1}_{R}\otimes\Theta\in\mathrm{FREE}\left(RA\to RB\right)$.

Note that appending free channels is a free operation: it is the tensor
product of the identity superchannel with a free channel. Therefore,
for any given free channel $\mathcal{M}_{B}\in\mathfrak{F}\left(B_{0}\to B_{1}\right)$,
the superchannel $\Theta_{\mathcal{M}}\left[\mathcal{N}_{A}\right]:=\mathcal{N}_{A}\otimes\mathcal{M}_{B}$
is a free superchannel, i.e.\ it belongs to $\mathrm{FREE}\left(A\to AB\right)$.

In some important resource theories, e.g.\ in entanglement theory
\citep{LOCC1,LOCC2,Lo-Popescu,Chitambar2014}, the set of natural
free operations can be hard to characterize mathematically \citep{NP-hard1,NP-hard2}.
For this reason, it can be convenient to enlarge the set of free operations
to work with a less complicated set. A standard enlargement is to
consider \emph{all} resource non-generating (RNG) superchannels \citep{Quantum-resource-2,Gour-review}:
\begin{align*}
 & {\rm RNG}\left(A\to B\right)\\
 & :=\left\{ \Theta\in\mathfrak{S}\left(A\to B\right):\Theta\left[\mathcal{M}_{A}\right]\in\mathfrak{F}\left(B_{0}\to B_{1}\right)\right\} ,
\end{align*}
for all $\mathcal{M}_{A}\in\mathfrak{F}\left(A_{0}\to A_{1}\right)$.
Similarly to the case of states, this is the set of superchannels
that transform free channels into free channels. In this setting,
since we require free superchannels to be completely free, RNG superchannels
are also \emph{completely resource non-generating} (CRNG) (in general,
however, they are two distinct sets, with $\mathrm{CRNG}\subseteq\mathrm{RNG}$):
$\Theta$ is CRNG if and only if $\mathbbm{1}_{R}\otimes\Theta$ is
RNG, for all systems $R=\left(R_{0},R_{1}\right)$. In Ref.~\citep{Dynamical-entanglement}
we consider PPT operations \citep{PPT1,PPT2} and separable operations
\citep{SEP} as extensions of the LOCC paradigm. Both of these sets
are CRNG. Note that, however, a priori, there is no guarantee that
CRNG superchannels are freely realizable in terms of CRNG channels
in the underlying resource theory of states.

Dynamical resources are quantified by dynamical resource monotones.
\begin{defn}
Let $\mathfrak{F}$ be a QRT admitting a tensor product structure.
Let $f:\mathrm{CPTP}\to\mathbb{R}$ be a function on the set of all
channels in all dimensions. Then, $f$ is called a \emph{dynamical
resource monotone} if, for every channel $\mathcal{N}\in\mathrm{CPTP}\left(A_{0}\to A_{1}\right)$
and every superchannel $\Theta\in\mathrm{FREE}\left(A\to B\right)$,
$f\left(\Theta\left[\mathcal{N}_{A}\right]\right)\leq f\left(\mathcal{N}_{A}\right).$
\end{defn}

It is customary, although not essential, to request that, for any
system $A_{0}$, the value of $f$ on the identity channel $\mathsf{id}_{A_{0}}$
is zero; i.e.\ $f\left(\mathsf{id}_{A_{0}}\right)=0$. This condition
implies that $f$ is non-negative, and satisfies 
\begin{equation}
f\left(\mathcal{N}_{A}\right)=0\quad\forall\mathcal{N}\in\mathfrak{F}\left(A_{0}\to A_{1}\right),\label{eq:zero monotone}
\end{equation}
for every system $A=\left(A_{0},A_{1}\right)$. The above property
follows from a combination of the monotonicity property of $f$ with
the fact that the replacement superchannel that takes any channel
to a fixed free channel is itself a free superchannel, as it can be
realized with free pre- and post-processing. Applying the replacement
superchannel preparing $\mathcal{N}\in\mathfrak{F}\left(A_{0}\to A_{1}\right)$
to the identity superchannel, we get $f\left(\mathcal{N}\right)\leq f\left(\mathsf{id}_{A_{0}}\right)=0$.
Applying the replacement superchannel preparing the identity channel
to $\mathcal{N}$ instead yields $f\left(\mathcal{N}\right)\geq f\left(\mathsf{id}_{A_{0}}\right)=0$,
whence Eq.~\eqref{eq:zero monotone} follows.

Examples of dynamical monotones that are given in terms of the relative
entropy were discussed in Refs.~\citep{Pirandola-LOCC,Resource-channels-1,Resource-channels-2,Gour-Winter,Wilde-entanglement}.
One such example is defined in terms of the \emph{channel divergence}
\citep{Cooney2016,Datta,Gour2018}. Given two channels $\mathcal{N},\mathcal{E}\in\mathrm{CPTP}\left(A_{0}\to A_{1}\right)$,
the channel divergence is
\[
D\left(\mathcal{N}_{A}\middle\|\mathcal{E}_{A}\right):=\sup_{\psi_{RA_{0}}}D\left(\mathcal{N}_{A}\left(\psi_{RA_{0}}\right)\middle\|\mathcal{E}_{A}\left(\psi_{RA_{0}}\right)\right)
\]
where $D\left(\rho\middle\|\sigma\right):=\mathrm{tr}\left[\rho\log\rho\right]-\mathrm{tr}\left[\rho\log\sigma\right]$
is the relative entropy, $R$ is a reference system, and the supremum
is over all $\left|R\right|$ and all density matrices $\psi_{RA_{0}}\in\mathfrak{D}\left(RA_{0}\right)$.
In Refs.~\citep{Cooney2016,Datta,Gour2018} it was argued that the
supremum can be replaced with a maximum, $R$ can be taken to have
the same dimension as $A_{0}$, and $\psi_{RA_{0}}$ can be taken
to be pure. The relative entropy of a dynamical resource $\mathcal{N}\in\mathrm{CPTP}\left(A_{0}\to A_{1}\right)$
is defined as
\[
D_{\mathfrak{F}}\left(\mathcal{N}_{A}\right):=\min_{\mathcal{E}\in\mathfrak{F}\left(A_{0}\to A_{1}\right)}D\left(\mathcal{N}_{A}\middle\|\mathcal{E}_{A}\right).
\]

There is also a way to elevate any static monotone into a dynamical
monotone. Given a static monotone $E$, define
\[
E\left(\mathcal{N}_{A}\right):=\sup_{\sigma\in\mathfrak{D}\left(RA_{0}\right)}E\left(\mathcal{N}_{A}\left(\sigma_{RA_{0}}\right)\right)-E\left(\sigma_{RA_{0}}\right),
\]
for any $\mathcal{N}\in\mathrm{CPTP}\left(A_{0}\to A_{1}\right)$.
Then, it can be shown that $E$ is non-increasing under CRNG superchannels
\citep{Resource-channels-1,Resource-channels-2,Gour-Winter}. This
was called \emph{amortized extension} in Ref.~\citep{Gour-Winter}.
This definition captures the generating power of the channel $\mathcal{N}$,
understood as the maximum amount of static resource $\mathcal{N}$
can generate.

\subsection{A complete family of dynamical monotones\label{subsec:mono}}

The examples of dynamical monotones presented in the previous subsection
are typically very hard to compute due to the optimizations involved.
Here for the first time we introduce a family of dynamical resource
montones for convex resource theories that in some cases (e.g.\ NPT
entanglement, Ref.~\citep{Dynamical-entanglement}) can be computed
with SDPs. Furthermore, each member of the family is convex, and the
family itself is \emph{complete}, in the sense that the monotones
provide both \emph{necessary and sufficient} conditions for the conversion
of a dynamical resource into another with free superchannels. In this
sense, this family of monotones fully captures the resourcefulness
of a dynamical resource.

An example of a complete family of static resource monotones is known
for pure-state entanglement theory \citep{Nielsen,Vidal,Jonathan}.
There, the family of entanglement monotones is given in terms of Ky-Fan
norms, and due to Nielsen majorization theorem \citep{Nielsen}, this
family provides both necessary and sufficient conditions for the convertibility
of pure bipartite states. The fact that the family consists of a \emph{finite}
number of monotones makes it easy to determine the convertibility
of bipartite pure states under LOCC. However, for mixed states it
is known that, already in local dimension 4, a finite number of monotones
is insufficient to fully determine the exact interconversions between
bipartite mixed states \citep{Gour-infinite}. Therefore, in general,
one cannot expect to find a finite and complete family for a generic
QRT. 
\begin{thm}
\label{thm:monotones}Let $\mathrm{FREE}\left(A\to B\right)$ be as
above, such that for every two systems $A=\left(A_{0},A_{1}\right)$
and $B=\left(B_{0},B_{1}\right)$, the set $\mathrm{FREE}\left(A\to B\right)$
is \emph{convex} and topologically closed. For any quantum channel
$\mathcal{P}_{B}\in\mathrm{CPTP}\left(B_{0}\to B_{1}\right)$ define
\[
f_{\mathcal{P}}\left(\mathcal{N}_{A}\right):=\max_{\Theta\in\mathrm{FREE}\left(A\to B\right)}\left\langle \mathcal{P}_{B},\Theta\left[\mathcal{N}_{A}\right]\right\rangle ,
\]
for every $\mathcal{N}_{A}\in\mathrm{CPTP}\left(A_{0}\to A_{1}\right)$.
Let $\mathcal{N}_{A}\in\mathrm{CPTP}\left(A_{0}\to A_{1}\right)$
and $\mathcal{M}_{B}\in\mathrm{CPTP}\left(B_{0}\to B_{1}\right)$
be two quantum channels. Then, $\mathcal{M}_{B}=\Theta_{A\to B}\left[\mathcal{N}_{A}\right]$,
for some superchannel $\Theta\in\mathrm{FREE}\left(A\to B\right)$
if and only if 
\[
f_{\mathcal{P}}\left(\mathcal{N}_{A}\right)\geq f_{\mathcal{P}}\left(\mathcal{M}_{B}\right)\quad\forall\mathcal{P}\in\mathrm{CPTP}\left(B_{0}\to B_{1}\right).
\]
\end{thm}

\begin{proof}
Denote
\[
\mathfrak{C}_{\mathcal{N}}:=\left\{ \Theta\left[\mathcal{N}\right]:\Theta\in\mathrm{FREE}\left(A\to B\right)\right\} .
\]
Since we assume that $\mathrm{FREE}$ is convex and closed, so is
$\mathfrak{C}_{\mathcal{N}}$. Therefore, by the supporting hyperplane
theorem, $\mathcal{M}_{B}\not\in\mathfrak{C}_{\mathcal{N}}$ if and
only if there exists a Hermitian-preserving map $\mathcal{P}_{B}\in\mathrm{Herm}\left(B_{0}\to B_{1}\right)$
such that
\[
\left\langle \mathcal{P}_{B},\mathcal{M}_{B}\right\rangle >\max_{\Theta\in\mathrm{FREE}\left(A\to B\right)}\left\langle \mathcal{P}_{B},\Theta\left[\mathcal{N}_{A}\right]\right\rangle .
\]
Alternatively, $\mathcal{M}_{B}\in\mathfrak{C}_{\mathcal{N}}$ if
and only if for all Hermitian-preserving maps $\mathcal{P}_{B}\in\mathrm{Herm}\left(B_{0}\to B_{1}\right)$
\begin{equation}
\left\langle \mathcal{P}_{B},\mathcal{M}_{B}\right\rangle \leq\max_{\Theta\in\mathrm{FREE}\left(A\to B\right)}\left\langle \mathcal{P}_{B},\Theta\left[\mathcal{N}_{A}\right]\right\rangle .\label{a1}
\end{equation}
First we show that the above inequality holds for all Hermitian-preserving
maps $\mathcal{P}\in\mathrm{Herm}\left(B_{0}\to B_{1}\right)$ if
and only if
\begin{align}
f_{\mathcal{P}}\left(\mathcal{M}_{B}\right)= & \max_{\Theta'\in\mathrm{FREE}\left(B\to B\right)}\left\langle \mathcal{P}_{B},\Theta'\left[\mathcal{M}_{B}\right]\right\rangle \nonumber \\
 & \leq\max_{\Theta\in\mathrm{FREE}\left(A\to B\right)}\left\langle \mathcal{P}_{B},\Theta\left[\mathcal{N}_{A}\right]\right\rangle \nonumber \\
 & =f_{\mathcal{P}}\left(\mathcal{N}_{A}\right)\label{a2}
\end{align}
for all Hermitian-preserving $\mathcal{P}_{B}\in\mathrm{Herm}\left(B_{0}\to B_{1}\right)$.
Indeed, if Eq.~\eqref{a2} holds, then take $\Theta'$ to be the
identity superchannel $\mathbbm{1}_{B}$; thus we immediately get
Eq.~\eqref{a1} because
\begin{align*}
\left\langle \mathcal{P}_{B},\mathcal{M}_{B}\right\rangle  & \leq\max_{\Theta'\in\mathrm{FREE}\left(B\to B\right)}\left\langle \mathcal{P}_{B},\Theta'\left[\mathcal{M}_{B}\right]\right\rangle \\
 & \leq\max_{\Theta\in\mathrm{FREE}\left(A\to B\right)}\left\langle \mathcal{P}_{B},\Theta\left[\mathcal{N}_{A}\right]\right\rangle .
\end{align*}
Conversely, suppose Eq.~\eqref{a1} holds. Then, for any $\Theta'\in\mathrm{FREE}\left(B\to B\right)$
we have
\begin{align*}
\left\langle \mathcal{P}_{B},\Theta'\left[\mathcal{M}_{B}\right]\right\rangle  & =\left\langle \Theta^{\prime*}\left[\mathcal{P}_{B}\right],\mathcal{M}_{B}\right\rangle \\
 & \leq\max_{\Theta\in\mathrm{FREE}\left(A\to B\right)}\left\langle \Theta^{\prime*}\left[\mathcal{P}_{B}\right],\Theta\left[\mathcal{N}_{A}\right]\right\rangle \\
 & =\max_{\Theta\in\mathrm{FREE}\left(A\to B\right)}\left\langle \mathcal{P}_{B},\left(\Theta'\circ\Theta\right)\left[\mathcal{N}_{A}\right]\right\rangle \\
 & \leq\max_{\Theta\in\mathrm{FREE}\left(A\to B\right)}\left\langle \mathcal{P}_{B},\Theta\left[\mathcal{N}_{A}\right]\right\rangle ,
\end{align*}
where the first inequality follows from assuming Eq.~\eqref{a1},
and the last inequality from the property that if $\Theta$ and $\Theta'$
are both free, then $\Theta'\circ\Theta$ is also free. Eq.~\eqref{a2}
immediately holds. 

It is left to show that it is sufficient to take $\mathcal{P}_{B}$
to be a CPTP map. To this end, it will be convenient to express the
inner products in terms of the Choi matrices. Now, for any Hermitian-preserving
map $\mathcal{P}_{B}$, consider a CPTP map $\widetilde{\mathcal{P}}_{B}$
whose Choi matrix is 
\[
J_{B}^{\widetilde{\mathcal{P}}}:=\left(1-\varepsilon\right)I_{B_{0}}\otimes u_{B_{1}}+\varepsilon\left(J_{B}^{\mathcal{P}}+\left(I_{B_{0}}-J_{B_{0}}^{\mathcal{P}}\right)\otimes u_{B_{1}}\right),
\]
where $\varepsilon>0$ is small enough so that $J_{B}^{\widetilde{\mathcal{P}}}\geq0$.
Note also that $J_{B_{0}}^{\widetilde{\mathcal{P}}}=I_{B_{0}}$ so
that $\widetilde{\mathcal{P}}_{B}$ is a quantum channel. Now, a key
observation is that, for any quantum channel $\mathcal{N}_{B}$, we
get
\begin{align*}
 & \left\langle \widetilde{\mathcal{P}}_{B},\mathcal{M}_{B}\right\rangle =\mathrm{tr}\left[J_{B}^{\widetilde{\mathcal{P}}}J_{B}^{\mathcal{M}}\right]\\
 & =\left(1-\varepsilon\right)\frac{\left|B_{0}\right|}{\left|B_{1}\right|}+\varepsilon\mathrm{tr}\left[J_{B}^{\mathcal{P}}J_{B}^{\mathcal{M}}\right]+\varepsilon\frac{\left|B_{0}\right|}{\left|B_{1}\right|}-\varepsilon\frac{1}{\left|B_{1}\right|}\mathrm{tr}\left[J_{B}^{\mathcal{P}}\right].
\end{align*}
Hence
\[
\left\langle \widetilde{\mathcal{P}}_{B},\mathcal{M}_{B}\right\rangle =\varepsilon\left\langle \mathcal{P}_{B},\mathcal{M}_{B}\right\rangle +c^{\mathcal{P}},
\]
where 
\[
c^{\mathcal{P}}:=\frac{1}{\left|B_{1}\right|}\left(\left|B_{0}\right|-\varepsilon\mathrm{tr}\left[J_{B}^{\mathcal{P}}\right]\right)
\]
is a constant depending only on $\mathcal{P}_{B}$. Therefore, Eqs.~\eqref{a1}
and \eqref{a2} hold for $\mathcal{P}_{B}$ if and only if they hold
for $\widetilde{\mathcal{P}}_{B}$. In other words, it is sufficient
to consider CPTP maps $\mathcal{P}_{B}$.
\end{proof}
\begin{rem}
The definition of the functions $f_{\mathcal{P}}$ makes them convex.
\end{rem}

\begin{rem}
Similar families of monotones have been given recently in Refs.~\citep{Gour-single-shot,Girard,Chitambar-2018b,Gour2018,Gour2018b}
for static resource theories, and in Ref.~\citep{Takagi-Regula}
in the context of channel discrimination tasks (see also the related
discussion in Ref.~\citep{Single-shot-new}). The monotones constructed
in theorem~\ref{thm:monotones} can be reduced to all the ones introduced
in Ref.~\citep{Gour-single-shot,Girard,Chitambar-2018b,Gour2018,Gour2018b,Takagi-Regula,Single-shot-new},
when restricting some of the input/output subsystems to be trivial
or classical.
\end{rem}

The functions $f_{\mathcal{P}}$ behave monotonically under free superchannels,
therefore also under superchannels that replace any input channel
with a fixed free channel. This in turn implies that, for every $\mathcal{P}$,
$f_{\mathcal{P}}$ take the same value on all free channels $\mathcal{N}\in\mathfrak{F}\left(A_{0}\to A_{1}\right)$:
if $\mathcal{N}\in\mathfrak{F}\left(A_{0}\to A_{1}\right)$ we have
\begin{align*}
f_{\mathcal{P}}\left(\mathcal{N}_{A}\right) & =\max_{\Theta\in\mathrm{FREE}\left(A\to B\right)}\left\langle \mathcal{P}_{B},\Theta\left[\mathcal{N}_{A}\right]\right\rangle \\
 & =\max_{\mathcal{M}\in\mathfrak{F}\left(B_{0}\to B_{1}\right)}\left\langle \mathcal{P}_{B},\mathcal{M}_{B}\right\rangle \\
 & \equiv g\left(\mathcal{P}_{B}\right).
\end{align*}
Therefore, if we want monotones that vanish on free channels, for
any $\mathcal{P}\in\mathrm{CPTP}\left(B_{0}\to B_{1}\right)$, define
\[
G_{\mathcal{P}}\left(\mathcal{N}_{A}\right):=f_{\mathcal{P}}\left(\mathcal{N}_{A}\right)-g\left(\mathcal{P}_{B}\right).
\]
In this way, $\left\{ G_{\mathcal{P}}\right\} $ is a complete set
of non-negative resource monotones that vanish on free channels.

The way $f_{\mathcal{P}}$ were constructed means that they can be
expressed in terms of resource witnesses. To see why, denote the set
of (free) Choi matrices by 
\begin{equation}
\mathfrak{J}_{AB}:=\left\{ \mathbf{J}_{AB}^{\Theta}:\Theta\in\mathrm{FREE}\left(A\to B\right)\right\} .\label{eq:Jset}
\end{equation}
Since $\mathrm{FREE}\left(A\to B\right)$ is closed and convex, so
is $\mathfrak{J}_{AB}$. The monotones $f_{\mathcal{P}}$ can be expressed
as
\begin{align*}
f_{\mathcal{P}}\left(\mathcal{N}_{A}\right) & =\max_{\Theta\in\mathrm{FREE}\left(A\to B\right)}\left\langle \mathcal{P}_{B},\Theta\left[\mathcal{N}_{A}\right]\right\rangle \\
 & =\max_{\Theta\in\mathrm{FREE}\left(A\to B\right)}\mathrm{tr}\left[J_{B}^{\mathcal{P}}J_{B}^{\Theta\left[\mathcal{N}_{A}\right]}\right]\\
 & =\max_{\Theta\in\mathrm{FREE}\left(A\to B\right)}\mathrm{tr}\left[\mathbf{J}_{AB}^{\Theta}\left(\left(J_{A}^{\mathcal{N}}\right)^{T}\otimes J_{B}^{\mathcal{P}}\right)\right]\\
 & =\max_{\mathbf{J}_{AB}\in\mathfrak{J}_{AB}}\mathrm{tr}\left[\mathbf{J}_{AB}\left(\left(J_{A}^{\mathcal{N}}\right)^{T}\otimes J_{B}^{\mathcal{P}}\right)\right].
\end{align*}
In other terms, $f_{\mathcal{P}}\left(\mathcal{N}_{A}\right)$ is
the support function of $\mathfrak{J}_{AB}$ evaluated at $\left(J_{A}^{\mathcal{N}}\right)^{T}\otimes J_{B}^{\mathcal{P}}$.
Let $\mathfrak{K}$ be the (convex) cone obtained from $\mathfrak{J}_{AB}$
by multiplying its elements by a non-negative number, i.e.\ $\mathfrak{K}:=\mathbb{R}_{+}\mathfrak{J}_{AB}$.
With this definition we can write
\[
f_{\mathcal{P}}\left(\mathcal{N}_{A}\right)=\max_{\substack{\mathbf{J}_{AB}\in\mathfrak{K}\\
\mathrm{tr}\left[\mathbf{J}_{AB}\right]=\left|A_{1}B_{0}\right|
}
}\mathrm{tr}\left[\mathbf{J}_{AB}\left(\left(J_{A}^{\mathcal{N}}\right)^{T}\otimes J_{B}^{\mathcal{P}}\right)\right].
\]
The above optimization problem is a conic linear program. As such,
using duality, $f_{\mathcal{P}}$ can be equivalently expressed as
\[
f_{\mathcal{P}}\left(\mathcal{N}_{A}\right)=\left|A_{1}B_{0}\right|\min\left\{ x:xI_{AB}-\left(J_{A}^{\mathcal{N}}\right)^{T}\otimes J_{B}^{\mathcal{P}}\in\mathfrak{K}^{*}\right\} ,
\]
where $x\in\mathbb{R}$ and $\mathfrak{K}^{*}$ is the dual cone
\begin{equation}
\mathfrak{K}^{*}=\left\{ W\in\mathrm{Herm}\left(AB\right):\mathrm{tr}\left[WM\right]\geq0\quad\forall M\in\mathfrak{K}\right\} .\label{eq:dual cone}
\end{equation}
Since the cone $\mathfrak{K}$ consists of only positive semi-definite
matrices, it follows that any positive semi-definite matrix belongs
to $\mathfrak{K}^{*}$. Note also that we must have $x>0$ in the
equation above, otherwise $M:=xI_{AB}-\left(J_{A}^{\mathcal{N}}\right)^{T}\otimes J_{B}^{\mathcal{P}}<0$,
and therefore $M$ would not belong to $\mathfrak{K}^{*}$. 

The cone $\mathfrak{K}$ is convex and closed. Therefore, as a consequence
of the hyperplane separation theorem, $\mathfrak{K}^{**}=\mathfrak{K}$.
This in particular implies that $M\in\mathfrak{K}$ if and only if
$\mathrm{tr}\left[MW\right]\geq0$ for all $W\in\mathfrak{K}^{*}$.
Hence the Hermitian matrices (observables) in $\mathfrak{K}^{*}$
that are \emph{not} positive semi-definite can be viewed as witnesses
of supermaps that are not free. However, among them, some will only
witness whether or not a matrix $M$ corresponds to a valid superchannel,
while others will witness if it corresponds to a non-free superchannel.

\subsection{Single-shot interconversions with conic linear programming}

Here we consider single-shot interconversions between resources. For
this purpose, following similar ideas to \citep{Tomamichel}, we define
the \emph{conversion distance} for any two channels $\mathcal{N}_{A}$
and $\mathcal{M}_{B}$ as
\[
d_{\mathfrak{F}}\left(\mathcal{N}_{A}\to\mathcal{M}_{B}\right):=\frac{1}{2}\min_{\Theta\in\mathrm{FREE}\left(A\to B\right)}\left\Vert \Theta_{A\to B}\left[\mathcal{N}_{A}\right]-\mathcal{M}_{B}\right\Vert _{\diamond}.
\]
If $d_{\mathfrak{F}}\left(\mathcal{N}_{A}\to\mathcal{M}_{B}\right)\leq\varepsilon$,
for some small $\varepsilon>0$, we will say that $\mathcal{N}_{A}$
can be converted to $\mathcal{M}_{B}$ by free superchannels up to
a small error $\varepsilon$. When $\varepsilon=0$, the conversion
is exact.

Therefore, calculating the conversion distance between two channels
becomes equivalent to determining if the former channel can be converted
into the latter. As such, an important question is whether the conversion
distance can be computed efficiently. First of all, recall that, as
far as the diamond norm is concerned, the answer is positive, because
it can be expressed as the SDP \citep{Watrous-diamond}
\[
\frac{1}{2}\left\Vert \mathcal{E}_{B}-\mathcal{F}_{B}\right\Vert _{\diamond}=\min_{\omega_{B}\geq0;\thinspace\omega_{B}\geq J_{B}^{\mathcal{E}-\mathcal{F}}}\left\Vert \omega_{B_{0}}\right\Vert _{\infty},
\]
for all $\mathcal{E},\mathcal{F}\in\mathrm{CPTP}\left(B_{0}\to B_{1}\right)$.
Now, in Ref.~\citep{Gour-Winter}, it was shown that it can be written
also as 
\[
\frac{1}{2}\left\Vert \mathcal{E}_{B}-\mathcal{F}_{B}\right\Vert _{\diamond}=\min\left\{ \lambda:\lambda\mathcal{Q}_{B}\geq\mathcal{E}_{B}-\mathcal{F}_{B}\right\} ,
\]
where $\mathcal{Q}_{B}\in\mathrm{CPTP}\left(B_{0}\rightarrow B_{1}\right)$.
Now take $\mathcal{E}_{B}:=\Theta_{A\to B}\left[\mathcal{N}_{A}\right]$
and $\mathcal{F}_{B}:=\mathcal{M}_{B}$; $d_{\mathfrak{F}}\left(\mathcal{N}_{A}\to\mathcal{M}_{B}\right)$
becomes
\begin{equation}
d_{\mathfrak{F}}\left(\mathcal{N}_{A}\to\mathcal{M}_{B}\right)=\min\left\{ \lambda:\lambda\mathcal{Q}_{B}\geq\Theta_{A\to B}\left[\mathcal{N}_{A}\right]-\mathcal{M}_{B}\right\} ,\label{singleshot}
\end{equation}
where $\Theta\in\mathrm{FREE}\left(A\to B\right)$ and $\mathcal{Q}\in\mathrm{CPTP}\left(B_{0}\to B_{1}\right)$.
This can be phrased as a conic linear program, so it has a dual (see
appendix~\ref{app:dual} for details), by which $d_{\mathfrak{F}}\left(\mathcal{N}_{A}\to\mathcal{M}_{B}\right)$
can also be expressed as
\begin{equation}
d_{\mathfrak{F}}\left(\mathcal{N}_{A}\to\mathcal{M}_{B}\right)=\max\left\{ t\left|A_{1}B_{0}\right|+\mathrm{tr}\left[\zeta_{B}J_{B}^{\mathcal{M}}\right]\right\} ,\label{ss2}
\end{equation}
subject to the constraints $0\leq\zeta_{B}\leq\eta_{B_{0}}\otimes I_{B_{1}}$,
$\mathrm{tr}\left[\eta_{B_{0}}\right]=1$, and $\left(J_{A}^{\mathcal{N}}\right)^{T}\otimes\zeta_{B}-tI_{AB}\in\mathfrak{K}^{*}$,
where the cone $\mathfrak{K}^{*}$ is the dual of the cone generated
by the Choi matrices of free superchannels (see Eq.~\eqref{eq:dual cone}).
If this cone has a simple characterization, as it happens e.g.\ in
NPT entanglement \citep{Dynamical-entanglement}, the problem of computing
$d_{\mathfrak{F}}\left(\mathcal{N}_{A}\to\mathcal{M}_{B}\right)$
becomes solving an SDP. However, for LOCC entanglement, determining
whether or not a superchannel is free is in general NP-hard \citep{NP-hard1,NP-hard2},
and consequently, so is the computation of $d_{\mathfrak{F}}\left(\mathcal{N}_{A}\to\mathcal{M}_{B}\right)$.

In entanglement theory there is a maximally entangled state \citep{Plenio-review,Review-entanglement},
which, with high enough dimension, can be converted to all other static
and dynamical resources by LOCC. An analogous situation also occurs,
e.g.\ in the resource theories of coherence \citep{Review-coherence}
and of purity \citep{Review-purity}, but, in general, not in all
resource theories. Such a maximal resource is most desirable and,
consequently, it is natural to consider the task of distilling a maximal
resource (i.e.\ resource distillation) and the task of forming a
resource from such a ``golden'' resource (i.e.\ resource cost)
\citep{Resource-currencies,Gour-review,Single-shot-new}. 

More precisely, let $\Phi_{B}^{+}$ be such a maximal resource on
system $B$, and fix $\varepsilon>0$. In the single-shot regime,
the \emph{$\varepsilon$-resource cost} of a channel $\mathcal{N}\in\mathrm{CPTP}\left(A_{0}\to A_{1}\right)$
is defined as
\begin{align}
 & \mathrm{COST}_{\mathfrak{F},\varepsilon}^{\left(1\right)}\left(\mathcal{N}_{A}\right)\nonumber \\
 & :=\log_{2}\min\left\{ \left|B\right|:d_{\mathfrak{F}}\left(\Phi_{B}^{+}\to\mathcal{N}_{A}\right)\leq\varepsilon\right\} \nonumber \\
 & =\log_{2}\min\left\{ \left|B\right|:\varepsilon\mathcal{Q}_{A}\geq\Theta_{B\to A}\left[\Phi_{B}^{+}\right]-\mathcal{N}_{A}\right\} ,\label{eq:cost1}
\end{align}
where $\Theta\in\mathrm{FREE}\left(B\to A\right)$ and $\mathcal{Q}\in\mathrm{CPTP}\left(A_{0}\to A_{1}\right)$.
The second equality in Eq.~\eqref{eq:cost1} follows from Eq.~\eqref{singleshot}.
The \emph{$\varepsilon$-resource distillation} of a channel $\mathcal{N}\in\mathrm{CPTP}\left(A_{0}\to A_{1}\right)$
is, instead, defined as
\begin{align*}
 & \mathrm{DISTILL}_{\mathfrak{F},\varepsilon}^{\left(1\right)}\left(\mathcal{N}_{A}\right)\\
 & :=\log_{2}\max\left\{ \left|B\right|:d_{\mathfrak{F}}\left(\mathcal{N}_{A}\to\Phi_{B}^{+}\right)\leq\varepsilon\right\} \\
 & =\log\max\left\{ \left|B\right|:\varepsilon\mathcal{Q}_{B}\geq\Theta_{A\to B}\left[\mathcal{N}_{A}\right]-\Phi_{B}^{+}\right\} ,
\end{align*}
where, again, $\Theta\in\mathrm{FREE}\left(A\to B\right)$, $\mathcal{Q}\in\mathrm{CPTP}\left(B_{0}\to B_{1}\right)$,
and the second equality follows from Eq.~\eqref{singleshot}.

\subsection{Definitions of various rates in the asymptotic regime\label{subsec:asymptotic}}

In the asymptotic regime, we are interested in the asymptotic rates
of converting one resource into another by means of the set of free
superchannels. The asymptotic\emph{ rate of conversion} from a channel
$\mathcal{N}\in\mathrm{CPTP}\left(A_{0}\to A_{1}\right)$ to a channel
$\mathcal{M}\in\mathrm{CPTP}\left(B_{0}\to B_{1}\right)$ is defined
as
\begin{align*}
 & R_{\mathfrak{F}}\left(\mathcal{N}_{A}\to\mathcal{M}_{B}\right)\\
 & :=\lim_{\varepsilon\to0^{+}}\inf\left\{ \frac{n}{m}:d_{\mathfrak{F}}\left(\mathcal{N}_{A}^{\otimes n}\to\mathcal{M}_{B}^{\otimes m}\right)\leq\varepsilon;\,m,n\in\mathbb{N}\right\} .
\end{align*}
If a maximal resource exists, we can also define the asymptotic resource
\emph{cost} and \emph{distillation} (see also Ref.~\citep{Resource-channels-1})
respectively as
\[
\mathrm{COST}_{\mathfrak{F}}\left(\mathcal{N}_{A}\right):=\lim_{\varepsilon\to0^{+}}\liminf_{n}\frac{1}{n}\mathrm{COST}_{\mathfrak{F},\varepsilon}^{\left(1\right)}\left(\mathcal{N}_{A}^{\otimes n}\right)
\]
and
\[
\mathrm{DISTILL}_{\mathfrak{F}}\left(\mathcal{N}_{A}\right):=\lim_{\varepsilon\to0^{+}}\limsup_{n}\frac{1}{n}\mathrm{DISTILL}_{\mathfrak{F},\varepsilon}^{\left(1\right)}\left(\mathcal{N}_{A}^{\otimes n}\right).
\]
Finally, one can also define the \emph{exact} resource cost and distillation
respectively as 
\[
\mathrm{COST}_{\mathfrak{F}}^{{\rm exact}}\left(\mathcal{N}_{A}\right):=\liminf_{n}\frac{1}{n}\mathrm{COST}_{\mathfrak{F},\varepsilon=0}^{\left(1\right)}\left(\mathcal{N}_{A}^{\otimes n}\right)
\]
and
\[
\mathrm{DISTILL}_{\mathfrak{F}}^{{\rm exact}}\left(\mathcal{N}_{A}\right):=\limsup_{n}\frac{1}{n}\mathrm{DISTILL}_{\mathfrak{F},\varepsilon=0}^{\left(1\right)}\left(\mathcal{N}_{A}^{\otimes n}\right).
\]
All the quantities above are typically very hard to compute. These
definitions mirror the analogous ones in resource theories of states
\citep{Quantum-resource-1,Resource-currencies,EastThesis,Gour-review,Kuroiwa2020generalquantum},
in which $n$ copies of a resource $\mathcal{N}$ are given in parallel,
and therefore they are described by the tensor product $\mathcal{N}^{\otimes n}$.
This is the only possibility for multiple static resources, which
can only be composed in parallel, i.e.\ with tensor product.

However, with dynamical resources, the time in which they are applied
starts playing a role. This is because dynamical resources have a
natural temporal ordering between input and output, and therefore
they can also be composed in non-parallel ways, e.g.\ in sequence.
Therefore when manipulating dynamical resources, it is not enough
to specify the CPTP maps involved but also \emph{when} (and how) they
can be used (see also Ref.~\citep{Resource-channels-2}). This opens
up the possibility of using \emph{adaptive schemes} when we have several
dynamical resources \citep{Adaptive-metrology,Pirandola-adaptive-metrology,Pirandola-LOCC,Kaur2017,Wilde-cost,Wilde-entanglement,Dynamical-entanglement}.
For example, if we have $n$ resources $\mathcal{N}_{1},\dots,\mathcal{N}_{n}$
that are available, respectively, at times $t_{1}\leq t_{2}\leq\dots\leq t_{n}$,
then the most general channel that can be simulated by free operations
using these resources is depicted in Fig.~\ref{comb}, where the
channels $\mathcal{E}_{1},\dots,\mathcal{E}_{n+1}$ are all free.
We use the notation $\mathscr{C}_{n}\left[\mathcal{N}_{1},\dots,\mathcal{N}_{n}\right]$
to describe the resulting channel, and $\mathscr{C}_{n}\left[\mathcal{N}^{n}\right]:=\mathscr{C}_{n}\left[\mathcal{N},\dots,\mathcal{N}\right]$
when all the resources $\mathcal{N}_{1},\dots,\mathcal{N}_{n}$ are
the same, and equal to $\mathcal{N}$. Note that this scheme includes
the two cases in which the $n$ resources are composed in parallel
(i.e.\ $\mathcal{N}_{1}\otimes\dots\otimes\mathcal{N}_{n}$) and
in sequence (i.e.\ $\mathcal{N}_{n}\circ\dots\circ\mathcal{N}_{1}$).
\begin{rem}
Note that, since the $n$ slots of a quantum comb are causally ordered
\citep{Circuit-architecture,Hierarchy-combs,Switch}, it is important
to know the order in which the resources are inserted. If the $n$
channels $\mathcal{N}_{1},\dots,\mathcal{N}_{n}$ are all available
at the initial time, and we do not know which to plug first into the
comb, then we must pick a particular ordering of them. More formally,
we need to pick a permutation $\pi\in S_{n}$ that fixes the causal
ordering between the $n$ resources, whereby their most general manipulation
is $\mathscr{C}_{n}\left[\mathcal{N}_{\pi\left(1\right)},\dots,\mathcal{N}_{\pi\left(n\right)}\right]$.
\end{rem}

With this in mind, when a maximal resource exists, we define the single-shot
\emph{adaptive} $\varepsilon$-resource\emph{ cost} of a channel $\mathcal{N}\in\mathrm{CPTP}\left(A_{0}\to A_{1}\right)$
as
\begin{align*}
 & \mathrm{COST}_{\mathfrak{F},\varepsilon}^{\left(n\right),{\rm Ad}}\left(\mathcal{N}_{A}\right)\\
 & :=\log\min\left\{ \left|B\right|^{n}:d_{\mathfrak{F}}\left(\mathscr{C}_{n}\left[\Phi_{B}^{+n}\right]\to\mathcal{N}_{A}\right)\leq\varepsilon\right\} .
\end{align*}
The single-shot \emph{adaptive} $\varepsilon$-resource \emph{distillation}
of a channel $\mathcal{N}\in\mathrm{CPTP}\left(A_{0}\to A_{1}\right)$
is, instead, defined as
\begin{align*}
 & \mathrm{DISTILL}_{\mathfrak{F},\varepsilon}^{\left(n\right),{\rm Ad}}\left(\mathcal{N}_{A}\right)\\
 & :=\log\max\left\{ \left|B\right|:d_{\mathfrak{F}}\left(\mathscr{C}_{n}\left[\mathcal{N}_{A}^{n}\right]\to\Phi_{B}^{+}\right)\leq\varepsilon\right\} .
\end{align*}

The \emph{asymptotic adaptive rate} of conversion from a channel $\mathcal{N}\in\mathrm{CPTP}\left(A_{0}\to A_{1}\right)$
to a channel $\mathcal{M}\in\mathrm{CPTP}\left(B_{0}\to B_{1}\right)$
by free operations is given by
\begin{align*}
 & R_{\mathfrak{F}}^{{\rm Ad}}\left(\mathcal{N}_{A}\to\mathcal{M}_{B}\right)\\
 & :=\lim_{\varepsilon\to0^{+}}\inf\left\{ \frac{n}{m}:d_{\mathfrak{F}}\left(\mathscr{C}_{n}\left[\mathcal{N}^{n}\right]\to\mathcal{M}_{B}^{m}\right)\leq\varepsilon;\thinspace m,n\in\mathbb{N}\right\} .
\end{align*}
Here by $\mathcal{M}_{B}^{m}$ we denote the channel $\mathscr{D}_{m}\left[\mathcal{M}_{B}^{m}\right]$,
i.e.\ the action of a (possibly non-free) comb $\mathscr{D}_{m}$
on $m$ copies of $\mathcal{M}_{B}$ inserted in its $m$ slots. Again,
this also includes the case in which the target resource $\mathcal{M}_{B}$
arises in $m$ parallel copies, i.e.\ $\mathcal{M}_{B}^{\otimes m}$.
If a maximal resource exists, we can also define the \emph{asymptotic
adaptive resource cost} and \emph{adaptive resource distillation}
respectively as
\[
\mathrm{COST}_{\mathfrak{F}}^{{\rm Ad}}\left(\mathcal{N}_{A}\right):=\lim_{\varepsilon\to0^{+}}\liminf_{n}\frac{1}{m}\mathrm{COST}_{\mathfrak{F},\varepsilon}^{\left(n\right),{\rm Ad}}\left(\mathcal{N}_{A}^{m}\right)
\]
and
\[
\mathrm{DISTILL}_{\mathfrak{F}}^{{\rm Ad}}\left(\mathcal{N}_{A}\right):=\lim_{\varepsilon\to0^{+}}\lim_{n\to+\infty}\frac{1}{n}\mathrm{DISTILL}_{\mathfrak{F},\varepsilon}^{\left(n\right),{\rm Ad}}\left(\mathcal{N}_{A}\right),
\]
where, as above $\mathcal{N}_{A}^{m}$ denotes the action of a (possibly
non-free) comb on $m$ copies of $\mathcal{N}_{A}$ (note that $m$
depends on $n$). The adaptive \emph{exact} resource distillation
and resource cost are defined similarly as above.

\section{Conclusions\label{sec:Conclusions}}

In this article we presented the general framework for resource theories
of quantum processes. In particular, we introduced a new construction
of a complete family of monotones governing the simulation of channels
by free superchannels, which is valid in all \emph{convex} resource
theories of quantum processes. We showed that the problem of resource
interconversion can be turned into a conic linear program, whose hardness
depends on the particular resource theory under consideration.

Moreover, we also showed that shifting our focus from states to processes
introduces a richer landscape of protocols that can be implemented
for resource conversions. This stems from the fact that channels,
unlike states, have an input and an output, therefore they can be
composed in a variety of ways. Hence the most general manipulation
of multiple copies of a resource follows an \emph{adaptive schem}e,
in which the various copies are inserted into the slots of a free
circuit (a free comb). This scheme is most general, as it includes
the well-known case of the tensor product of many copies. This added
layer of complexity makes resource theories of processes far more
complicated to study than resource theories of states. 

However, this is not the only extra complication. A further difficulty
concerns the realization of free superchannels. We saw that a priori
there is no guarantee that all free superchannels are also freely
realizable, i.e.\ they can be implemented with free pre-processing
and post-processing channels. We conjecture that in fact there exist
free superchannels that are not freely realizable. In general, it
is hard to determine if a given free superchannel admits a free realization,
so we were not able to provide a conclusive answer to this issue.
However, the results we announced in Ref.~\citep{Dynamical-entanglement}
suggest that focusing only on freely realizable superchannels makes
the issue of studying resource interconversion much more complicated
than considering generic free superchannels.

A further possible generalization is to relax the hypothesis of causal
manipulation of multiple dynamical resources. Indeed, when multiple
resources are plugged into a free quantum comb to be converted, the
order in which they are inserted matters, for the slots of the comb
are causally ordered, and a resource cannot be used to ``influence''
the others that causally precede it. In this case, not restricting
to combs, but also considering superpositions of causal orders in
resource processing \citep{Process-matrix,Giarmatzi} might help us
get an advantage on resource manipulation, as we already know this
to happen in the case of the quantum switch \citep{Switch,Chiribella-zero,Chiribella-communication,chiribella2019quantum,Indefinite-metrology}.
\begin{acknowledgments}
G.\ G.\ would like to thank Francesco Buscemi, Eric Chitambar, Mark
Wilde, and Andreas Winter for many useful discussions related to the
topic of this paper. The authors acknowledge support from the Natural
Sciences and Engineering Research Council of Canada (NSERC) through
grant RGPIN-2020-03938, from the Pacific Institute for the Mathematical
Sciences (PIMS), and a from Faculty of Science Grand Challenge award
at the University of Calgary.
\end{acknowledgments}

\bibliographystyle{apsrev4-2}
\bibliography{PPT}

\onecolumngrid

\appendix

\section{How to calculate the dual of \textmd{\textup{\normalsize{}$d_{\mathfrak{F}}\left(\mathcal{N}\to\mathcal{M}\right)$\label{app:dual}}}}

The first step to determine the dual of the conic linear program associated
with Eq.~\eqref{singleshot} is to express $d_{\mathfrak{F}}\left(\mathcal{N}\to\mathcal{M}\right)$
using Choi matrices and the characterization of the diamond norm in
Ref.~\citep{Watrous-diamond}. We have
\begin{align}
d_{\mathfrak{F}}\left(\mathcal{N}_{A}\to\mathcal{M}_{B}\right) & =\min\left\{ \lambda:\omega_{B}\geq\mathrm{tr}_{A}\left[\mathbf{J}_{AB}^{\Theta}\left(\left(J_{A}^{\mathcal{N}}\right)^{T}\otimes I_{B}\right)\right]-J_{B}^{\mathcal{M}};\thinspace\lambda I_{B_{0}}\geq\omega_{B_{0}};\thinspace\Theta\in\mathrm{FREE}\left(A\to B\right);\thinspace\omega_{B}\geq0\right\} \nonumber \\
 & =\min\left\{ \lambda:\omega_{B}\geq\mathrm{tr}_{A}\left[\alpha_{AB}\left(\left(J_{A}^{\mathcal{N}}\right)^{T}\otimes I_{B}\right)\right]-J_{B}^{\mathcal{M}};\thinspace\lambda I_{B_{0}}\geq\omega_{B_{0}};\thinspace\omega_{B}\geq0;\thinspace\alpha_{AB}\in\mathfrak{J}_{AB}\right\} ,\label{eq:d_f 1}
\end{align}
where $\mathfrak{J}_{AB}$ is the set of the Choi matrices of free
superchannels (cf.\ Eq.~\eqref{eq:Jset}). We want to work with
the dual problem using conic linear programming, but $\mathfrak{J}_{AB}$,
albeit convex, is not a cone. Therefore we consider the cone $\mathfrak{K}$
generated by $\mathfrak{J}_{AB}$ (see subsection~\ref{subsec:mono}).
Now Eq.~\eqref{eq:d_f 1} can be rewritten as
\begin{align*}
 & d_{\mathfrak{F}}\left(\mathcal{N}_{A}\to\mathcal{M}_{B}\right)\\
 & =\min\left\{ \lambda:\omega_{B}\geq\mathrm{tr}_{A}\left[\alpha_{AB}\left(\left(J_{A}^{\mathcal{N}}\right)^{T}\otimes I_{B}\right)\right]-J_{B}^{\mathcal{M}};\thinspace\lambda I_{B_{0}}\geq\omega_{B_{0}};\thinspace\omega_{B}\geq0;\thinspace\alpha_{AB}\in\mathfrak{K};\thinspace\mathrm{tr}\left[\alpha_{AB}\right]=\left|A_{1}B_{0}\right|\right\} .
\end{align*}
Now, following Ref.~\citep{Barvinok}, consider the two convex cones
\[
\mathfrak{K}_{1}:=\left\{ \left(\lambda,\omega_{B},\alpha_{AB}\right):\lambda\in\mathbb{R}_{+};\,\omega_{B}\geq0;\thinspace\alpha_{AB}\in\mathfrak{K}\right\} 
\]
\[
\mathfrak{K}_{2}:=\left\{ \left(R_{B_{0}},P_{B},0\right):R_{B_{0}}\geq0;\thinspace P_{B}\geq0\right\} .
\]
$\mathfrak{K}_{1}$ is a subset of the vector space $\mathbb{R}\oplus\mathrm{Herm}\left(B\right)\oplus\mathrm{Herm}\left(AB\right)$,
whereas $\mathfrak{K}_{2}$ is a subset of $\mathrm{Herm}\left(B_{0}\right)\oplus\mathrm{Herm}\left(B\right)\oplus\mathbb{R}$.
These two vector spaces carry an inner product. For $\mathbb{R}\oplus\mathrm{Herm}\left(B\right)\oplus\mathrm{Herm}\left(AB\right)$
it is
\[
\left\langle \left(\lambda,\omega_{B},\alpha_{AB}\right),\left(\lambda',\omega'_{B},\alpha'_{AB}\right)\right\rangle =\lambda\lambda'+\mathrm{tr}\left[\omega_{B}\omega'_{B}\right]+\mathrm{tr}\left[\alpha_{AB}\alpha'_{AB}\right];
\]
for $\mathrm{Herm}\left(B_{0}\right)\oplus\mathrm{Herm}\left(B\right)\oplus\mathbb{R}$
it is
\[
\left\langle \left(\eta_{B_{0}},\zeta_{B},t\right),\left(\eta'_{B_{0}},\zeta'_{B},t'\right)\right\rangle =\mathrm{tr}\left[\eta_{B_{0}}\eta'_{B_{0}}\right]+\mathrm{tr}\left[\zeta_{B}\zeta'_{B}\right]+tt'.
\]

Now consider the linear map $\mathcal{L}:\mathbb{R}\oplus\mathrm{Herm}\left(B\right)\oplus\mathrm{Herm}\left(AB\right)\rightarrow\mathrm{Herm}\left(B_{0}\right)\oplus\mathrm{Herm}\left(B\right)\oplus\mathbb{R}$.
Its action on a generic element $X=\left(\lambda,\omega_{B},\alpha_{AB}\right)$
of $\mathfrak{K}_{1}$ is
\[
\mathcal{L}\left(X\right):=\left(\lambda I_{B_{0}}-\omega_{B_{0}},\omega_{B}-\mathrm{tr}_{A}\left[\alpha_{AB}\left(\left(J_{A}^{\mathcal{N}}\right)^{T}\otimes I_{B}\right)\right],\mathrm{tr}\left[\alpha_{AB}\right]\right).
\]
Notice that this specifies $\mathcal{L}$ completely because $\mathfrak{K}_{1}$
spans the whole domain of $\mathcal{L}$. Now consider 
\[
H_{1}=\left(1,0_{B},0_{AB}\right)\in\mathbb{R}\oplus\mathrm{Herm}\left(B\right)\oplus\mathrm{Herm}\left(AB\right)
\]
and
\[
H_{2}=\left(0_{B_{0}},J_{B}^{\mathcal{M}},\left|A_{1}B_{0}\right|\right)\in\mathrm{Herm}\left(B_{0}\right)\oplus\mathrm{Herm}\left(B\right)\oplus\mathbb{R}.
\]
With this notation we can write \citep{Barvinok}
\begin{align*}
d_{\mathfrak{F}}\left(\mathcal{N}_{A}\to\mathcal{M}_{B}\right) & =\min\left\{ \left\langle X,H_{1}\right\rangle :\mathcal{L}\left(X\right)-H_{2}\in\mathfrak{K}_{2};\thinspace X\in\mathfrak{K}_{1}\right\} \\
 & =\max\left\{ \left\langle Y,H_{2}\right\rangle :H_{1}-\mathcal{L}^{*}\left(Y\right)\in\mathfrak{K}_{1}^{*};\thinspace Y\in\mathfrak{K}_{2}^{*}\right\} ,
\end{align*}
where the second equality follows from strong duality. We only need
to calculate $\mathcal{L}^{*}\left(Y\right)$, where $Y=\left(\eta_{B_{0}},\zeta_{B},t\right)$
is in $\mathrm{Herm}\left(B_{0}\right)\oplus\mathrm{Herm}\left(B\right)\oplus\mathbb{R}$.
We have
\[
\mathcal{L}^{*}\left(Y\right)=\left(\mathrm{tr}\left[\eta_{B_{0}}\right],\zeta_{B}-\eta_{B_{0}}\otimes I_{B_{1}},tI_{AB}-\left(J_{A}^{\mathcal{N}}\right)^{T}\otimes\zeta_{B}\right).
\]
Hence,
\begin{align*}
d_{\mathfrak{F}}\left(\mathcal{N}_{A}\to\mathcal{M}_{B}\right) & =\max\left\{ t\left|A_{1}B_{0}\right|+\mathrm{tr}\left[\zeta_{B}J_{B}^{\mathcal{M}}\right]:\mathrm{tr}\left[\eta_{B_{0}}\right]\leq1;\thinspace0\leq\zeta_{B}\leq\eta_{B_{0}}\otimes I_{B_{1}};\thinspace\left(J_{A}^{\mathcal{N}}\right)^{T}\otimes\zeta_{B}-tI_{AB}\in\mathfrak{K}^{*}\right\} \\
 & =\max\left\{ t\left|A_{1}B_{0}\right|+\mathrm{tr}\left[\zeta_{B}J_{B}^{\mathcal{M}}\right]:\mathrm{tr}\left[\eta_{B_{0}}\right]=1;\thinspace0\leq\zeta_{B}\leq\eta_{B_{0}}\otimes I_{B_{1}};\,\left(J_{A}^{\mathcal{N}}\right)^{T}\otimes\zeta_{B}-tI_{AB}\in\mathfrak{K}^{*}\right\} ,
\end{align*}
where $\mathfrak{K}^{*}$ is the dual of the cone generated by the
Choi matrices of free superchannels. We have obtained Eq.~\eqref{ss2}.
\end{document}